\documentclass[acmlarge]{acmart}

\usepackage{booktabs} 

\usepackage[ruled]{algorithm2e} 

\usepackage{booktabs}
\newcommand{\tabitem}{~~\llap{\textbullet}~~}

\usepackage[justification=centering]{caption}

\usepackage{longtable}
\usepackage{amsmath}
\usepackage{varioref}
 
\usepackage{enumitem}
\usepackage{comment}

\usepackage{arabtex}
\usepackage{utf8}

\usepackage{color}

\usepackage{fancyvrb}
\definecolor{pblue}{rgb}{0.13,0.13,1}
\definecolor{pgreen}{rgb}{0,0.5,0}
\definecolor{pred}{rgb}{0.9,0,0}
\definecolor{pgrey}{rgb}{0.46,0.45,0.48}

\usepackage{fancyvrb}
\usepackage{fixltx2e}
\usepackage{bera}

\usepackage{listings}
\usepackage{mathtools}
\usepackage{amsmath}
\usepackage[title,titletoc,toc]{appendix}

\usepackage{graphicx}
\usepackage{subfig}

\usepackage{caption}
\usepackage{comment}

\usepackage{graphicx}
\usepackage[utf8]{inputenc}
\usepackage[english]{babel}
 
\usepackage{color}
\usepackage{color,soul}
\DeclareMathAlphabet{\mycal}{OMS}{zplm}{m}{n}
 
 \usepackage{amsmath}
\usepackage{varioref}
 
\usepackage{enumitem}
\usepackage{comment}

\usepackage{arabtex}
\usepackage{utf8}

\usepackage{color}

\definecolor{pblue}{rgb}{0.13,0.13,1}
\definecolor{pgreen}{rgb}{0,0.5,0}
\definecolor{pred}{rgb}{0.9,0,0}
\definecolor{pgrey}{rgb}{0.46,0.45,0.48}

\usepackage{fancyvrb}
\usepackage{fixltx2e}
\usepackage{bera}

\usepackage{listings}
\usepackage{mathtools}

\usepackage[title,titletoc,toc]{appendix}

\usepackage{graphicx}
\usepackage{subfig}

\usepackage{caption}
\usepackage{comment}

\usepackage{graphicx}
\usepackage[utf8]{inputenc}
\usepackage[english]{babel}
 
\usepackage{color}
\usepackage{color,soul}
\DeclareMathAlphabet{\mycal}{OMS}{zplm}{m}{n}

\usepackage{mathptmx}

\usepackage{soul}

\newcommand\bcmdtab{\noindent\bgroup\tabcolsep=0pt%
  \begin{tabular}{@{}p{10pc}@{}p{20pc}@{}}}
\newcommand\ecmdtab{\end{tabular}\egroup}

\usepackage{mathptmx}

\acmJournal{tocl}
\acmVolume{9}
\acmNumber{4}
\acmArticle{39}
\acmYear{2017}
\acmMonth{6}

\setcopyright{acmcopyright}

\acmDOI{0000001.0000001}

\newcommand*{\origrightarrow}{}

\renewcommand*{\textrightarrow}{\fontfamily{cmr}\selectfont\origrightarrow}

\begin{document}

\title{Visualization of Constraint Handling Rules: Semantics and Applications} 
\author{Nada Sharaf}
\authornote{The corresponding author}
\email{nada.hamed@guc.edu.eg}
\author{Slim Abdennadher}
\email{slim.abdennadher@guc.edu.eg}
\affiliation{%
  \institution{The German University in Cairo}
  \department{Computer Science and Engineering}
  \city{Cairo}
  \country{Egypt}
}
\author{Thom Fr\"{u}hwirth}
\affiliation{%
  \institution{Ulm University}
  \city{Ulm}
  \country{Germany}}

\begin{abstract}
The work in the paper presents an animation extension ($CHR^{vis}$) to Constraint Handling Rules (\textsf{CHR}).
Visualizations have always helped programmers understand data and debug programs. A picture is worth a thousand words. It can help identify where a problem is or show how something works. It can even illustrate a relation that was not clear otherwise. 
$CHR^{vis}$ aims at embedding animation and visualization features into \textsf{CHR} programs.  It thus enables users, while executing programs, to have such executions animated. The paper aims at providing the operational semantics for $CHR^{vis}$. The correctness of $CHR^{vis}$ programs is also discussed. Some applications of the new extension are also introduced.
\end{abstract}

%
%


%
%



\maketitle

\section{Introduction}
Animation tools are considered as a basic construct of programming languages. They are used to visualize the execution of a program. They provide users with a simple and intuitive method to debug and trace programs.
This paper presents an extension to Constraint Handling Rules (\textsf{CHR}). The extension adds new visual features to \textsf{CHR}. 
It thus enables users to have executions of \textsf{CHR} programs animated.

CHR \cite{chrbook,Fru98} has evolved over the years into a general purpose language. Originally, it was proposed for writing constraint solvers. Due to its declarativity, it has, however, been used with different algorithms such as sorting algorithms, graph algorithms, ... etc. \textsf{CHR} lacked tracing and debugging tools. Users were only able to use SWI-Prolog's textual trace facility as shown in Figure \ref{fig:textualtrace} which is hard to follow especially with big programs.
\begin{figure}[!ht]
  \subfloat[Using the normal trace option]{\hbox{\hspace{0cm}\includegraphics[width=100mm]{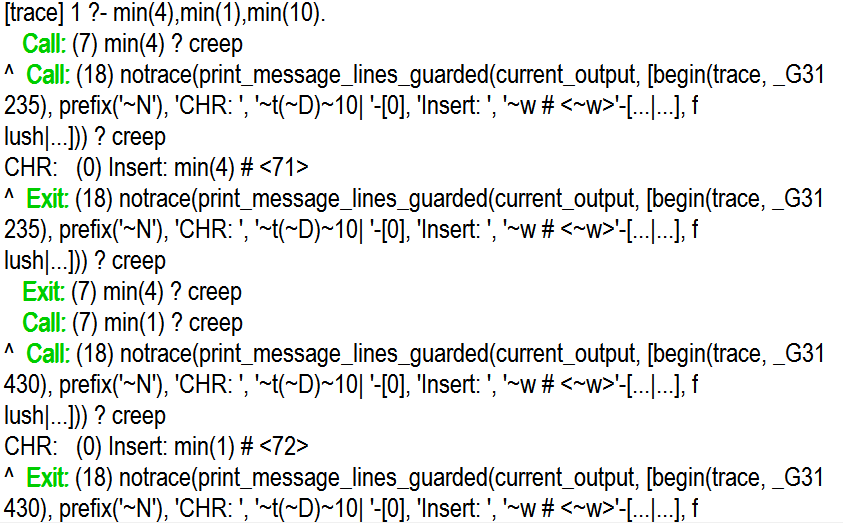}}}
  \\\subfloat[Using the chr\_trace option]{\hbox{\hspace{0cm}\includegraphics[width=100mm]{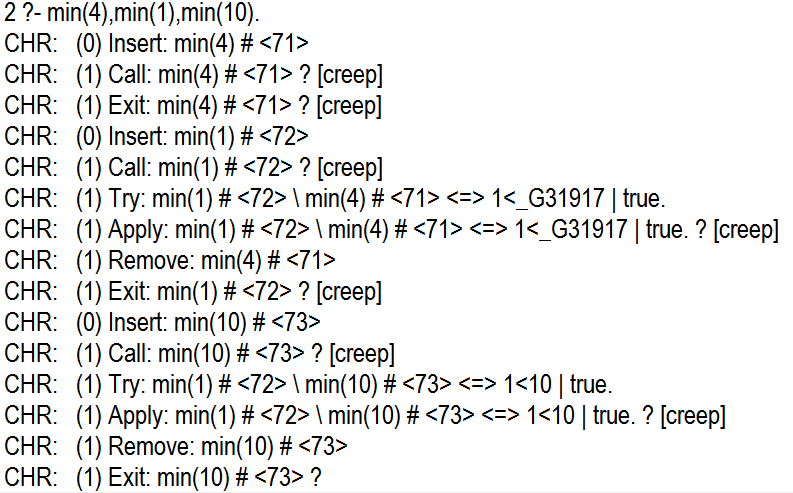}}}
  \caption{Current Tracing Facilities in SWI-Prolog.}
  \label{fig:textualtrace}
\end{figure}

{Two types of visual facilities are important for a \textsf{CHR} programmer/beginner. Firstly, the programmer would like to get a visual trace showing which \textsf{CHR} rule gets applied at every step and its effect.
Secondly, since \textsf{CHR} has developed into a general purpose language, it has been used with different types of algorithms such as sorting and graph algorithms. It is thus important to have a visual facility to animate the execution of the algorithms rather than just seeing the rules being executed.}
{\textsf{CHR} lacked such a tool. The tool should be able to adapt with the execution nature of \textsf{CHR} programs where constraints are added and removed continuously from the constraint store.}

Several approaches have been devised for visualizing \textsf{CHR} programs and its execution. In \cite{Abdennadher01avisualization}, a tool called \emph{VisualCHR} was proposed. VisualCHR allows its users to visually debug constraint solving. The compiler of JCHR \cite{schmauss_jack:jchr_1999} (on which VisualCHR is based) was modified. The visualization feature was thus not available for Prolog versions, the more prominent implementation of CHR.
\cite{DBLP:conf/iclp/AbdennadherS12} introduced a tool for visualizing the execution of \textsf{CHR} programs. It was able to show at every step the constraint store and the effect of applying each \textsf{CHR} rule in a step-by-step manner. The tool was based on the SWI-Prolog implementation of CHR. Source-to-source transformation was used in order to eliminate the need of doing any changes to the compiler. The tool could thus be deployed directly by any user. Source-to-source transformation is able to convert an original program automatically to a new one with the new features.

Despite of the availability of such visualization tools, \textsf{CHR} was still missing on a system for animating algorithms. The available tools were able to show at each point in time the executed rule and the status of the constraint store. However, the algorithm implemented had no effect on the produced visualization. 
Existing algorithm animation tools were either not generic enough or could not be adopted with \textsf{CHR}. For example, one of the available tools is XTANGO \cite{Stasko:xtango} which is a general purpose animating system supporting the development of algorithm
animations. However, the algorithm should be implemented in C or another language
such that it produces a trace file to be read by a C program driver making it difficult to use with \textsf{CHR}.
Due to the wide range of algorithms implemented through \textsf{CHR}, an algorithm-based animation was needed. Such animation should show at each step in time the changes to the data structure affected by the algorithm. 

{The paper presents a different direction for animating \textsf{CHR} programs.
It allows users to animate any kind of algorithm implemented in \textsf{CHR}.
This direction thus augments \textsf{CHR} with an animation extension.
As a result, it allows a \textsf{CHR} programmer to trace the program from an algorithmic point of view independent of the details of the execution of its rules.
The formal analysis of the new extension is presented in the paper. 
The paper thus presents a new operational semantics of \textsf{CHR} that embeds visualization into its execution.
The formalism is able to capture not only the behavior of the \textsf{CHR} rules, it is also able to represent the graphical objects associated with the animation. 
It is used to prove the correctness of the programs extended with animation features.
To eliminate the need of users learning the new syntax for using the extension, a transformation approach is also provided.
The correctness of the transformation approach is presented in the paper too.
}

The paper is organized as follows: Section \ref{sec:chr} introduces CHR. It starts with its syntax. It then shows two examples to show how \textsf{CHR} operates. The theoretical operational semantics is given in Section \ref{sec:theorsem}. The refined operational semantics is then introduced in  Section \ref{sec:refined}. Section \ref{sec:chrvis} introduces the new extension. The section starts by giving a general introduction. The syntax of annotation rules is then discussed and two examples are given. Finally, in Section \ref{sec:form} the formalization is given by introducing $\omega_{vis}$, a new operational semantics for \textsf{CHR} that accounts for annotation rules. The basic transitions include ones for simple constraint annotations. Afterwards, a more general semantics including rule annotations is presented. The separation was done to build the general semantics in an easy to follow method. Proofs of soundness and completeness are given.
Section \ref{sec:trans} presents a transformation approach that is able to transform normal \textsf{CHR} programs to $CHR^{vis}$ programs.
Section \ref{sec:app} introduces some applications that were built using $CHR^{vis}$. Conclusions and directions for future work are presented at the end of the paper.

\section{Constraint Handling Rules}
\label{sec:chr}
CHR was initially developed for writing constraint solvers \cite{chrbook,Fru98}.
The rules of a \textsf{CHR} program keeps on rewriting the constraints in the constraint store until a fixed point is reached. At that point no \textsf{CHR} rules could be applied. The constraint store is initialized by the constraints in the query of ths user.
\textsf{CHR} has implementations in different languages such as Java, C and Haskell. The most prominent implementation, however, is the Prolog one. A \textsf{CHR} program has two types of constraints: user-defined/\textsf{CHR} constraints and built-in constraint. \textsf{CHR} constraints are defined by the user at the beginning of a program. Built-in constraints, on the other hand, are handled by (the constraint theory $\mycal{CT}$) of the host language.

A \textsf{CHR} program consists of a set of, ``simpagation rules". A simpagation rule has the following format:
\begin{Verbatim}
            optional_rule_name @ H_k \ H_r <=> G | B.
\end{Verbatim}
\(H_{k}\) and \(H_{r}\) represent the head of the rule. The body of the rule is \(B\). The guard \(G\) represents a precondition for applying the rule. A rule is only applied if the constraint store contains constraints that match the head of the rule and if the guard is satisfied. As seen from the previous rule, the head has two parts: \(H_{k}\) and \(H_{r}\). The head of a rule could only contain \textsf{CHR} constraints. The guard should consist of built-in constraints. The body, on the other hand, can contain \textsf{CHR} and built-in constraints. On applying the rule, the constraints in \(H_{k}\) are kept in the constraint store. The constraints in \(H_{r}\) are removed from the constraint store. The body constraints are added to the constraint store.

There are two special kinds of \textsf{CHR} rules: propagation rules and simplification rules. A propagation rule has an empty \(H_{r}\) thus such a rule does not remove any constraint from the constraint store.
It has the following format:
\begin{Verbatim}[commandchars=\\\{\}]
            optional\_rule\_name @ H_{k}\; ==> G | B.
\end{Verbatim}
A simplification rule on the other hand has an empty \(H_{k}\). A simplification rule removes all the head constraints from the constraint store.
A simplification rule has the following format:
\begin{Verbatim}[commandchars=\\\{\}]
           optional\_rule\_name @ H_{r} <=> G | B. 
\end{Verbatim}
\subsection{Examples}
This section introduces examples of \textsf{CHR} programs. Because of the declarativity of \textsf{CHR} and its rules, it is possible to implement different algorithms with single-ruled programs. The first program shows a sorting algorithm. The second program (consisting of two rules) aims at finding the minimum number out of a set of numbers.

\subsubsection{Extracting the Minimum Number}
\label{subsec:extractingmin}
\(\)\\The following program extracts the minimum number out of a set of numbers.
Each number is represented by a \verb+min/1+ The rule \verb+get_min+ is a simpagation \textsf{CHR} rule. It is applied whenever there are two numbers \verb+X+ and \verb+Y+ such that \verb+X+ is less than \verb+Y+.
The rule removes the bigger number, \verb+Y+, from the constraint store. The smaller number, \verb+X+, is kept in the constraint store. Thus, similar to the previous program, when the rule is no longer applicable, a fixed point is reached. By then, the constraint store would only contain the smallest number.
The rule \verb+remove_dup+ is used to remove identical copies of the same number. The rule is thus fired if the user entered a sequence of numbers that contains duplicates.
\begin{verbatim}
:-chr_constraint min/1.
remove_dup @ min(X) \ min(X) <> true.
get_min @ min(X) \ min(Y) <=> X<Y | true.
\end{verbatim}
The following example shows the steps taken to find the minimum number out of the set: \verb+8+, \verb+3+ and \verb+1+.
\begin{center}
\(\underline{min\left(8\right),min\left(3\right)}, min\left(1\right)\)
\\\(\Downarrow\)\\
 \( \underline{min\left(3\right), min\left(1\right)}\)
\\\(\Downarrow\)\\
 \(min\left(1\right)\)
\end{center}

\subsubsection{Sorting an Array}
\label{sec:sorting}
\(\)\\The program aims at sorting numbers in an array/list. Each number is represented by the constraint \verb+cell(I,V)+. \verb+I+ represents the index and \verb+V+ represents the value of the element. The program contains one rule: \verb+sort_rule+. It is applied whenever the constraint store contains two \verb+cell+ constraints representing two unsorted elements. The guard makes sure that the two elements are not sorted with respect to each other. The element at index \verb+I1+ has a value (\verb+V1+) that is greater than the value (\verb+V2+) of the element at index \verb+I2+. \verb+I1+ is less than \verb+I2+. Thus, \verb+V1+ precedes \verb+V2+ in the array despite of the fact that it is greater than it. Since \verb+sort_rule+ is a simplification rule, the two constraints representing the unsorted elements are removed from the constraint store. Two \verb+cell+ constraints are added through the body of the rule to represent the performed swap to sort the two elements. The program is shown below:
\begin{verbatim}
:-chr_constraint cell/2.
sort_rule @ cell(In1,V1), cell(In2,V2) <=> In1<In2,V1>V2 | 
                                      cell(In2,V1), cell(In1,V2).
\end{verbatim}
Successive applications of the rule makes sure that any two elements that are not sorted with respect to each other are swapped. The fixed point is reached whenever \verb+sort_rule+ is no longer applicable. At this point, the array is sorted.
\\The following shows how the query \verb+cell(0,7), cell(1,6), cell(2,4)+ affects the constraint store:\footnote{The examples in this section are running with the refined operational semantics $\omega_{r}$ explained later. For simplicity, Some details are omitted.}
\begin{center}
\(\underline{cell\left(0,7\right), cell\left(1,6\right)}, cell\left(2,4\right)\)
\\\(\Downarrow\)\\
 \(cell\left((1,7\right), \underline{cell\left((0,6\right), cell\left((2,4\right)}\)
\\\(\Downarrow\)\\
 \(\underline{cell\left((1,7\right), cell\left((2,6\right)}, cell\left((0,4\right)\)
\\\(\Downarrow\)\\
 \(cell\left((2,7\right), cell\left((1,6\right), cell\left((0,4\right)\)
\end{center}

\subsection{Theoretical Operational Semantics $\omega_{t}$}
\label{sec:theorsem}
A \textsf{CHR} state is represented by the tuple $\langle G,S,B,T\rangle^{V}_{n}$ \cite{chrbook,TLP:7834577}. $G$ represents the goal store. It initially contains the query of the user. $S$ is the \textsf{CHR} constraint store containing the currently available \textsf{CHR} constraints. $B$, on the other hand, is the built-in store with the built-ins handled by the host language (Prolog in this case). The propagation history, $T$, holds the names of the applied \textsf{CHR} rules along with the IDs of the \textsf{CHR} constraints that activated the rules. $T$ is used to make sure a program does not run forever. Each \textsf{CHR} constraint is associated with an ID. $n$ represents the next available ID. $V$ represents the set of global variables. Such variables are the ones that exist in the initial query of the user. $V$ does not change during execution, it is thus omitted throughout the rest of the paper.
\\A variable $v \notin V$ is called a local variable \cite{raiserequiv}.

\begin{definition}
The function \verb+chr+ is defined such that \verb+chr(c#n)+ = \verb+c+. It is extended into sequences and sets of \textsf{CHR} constraints.  Likewise, the function \verb+id+ is defined such that \verb+id(c#n)+ = \verb+n+. It is also extended into sequences and sets of \textsf{CHR} constraints. 
\end{definition}

\begin{definition}
$vars\left(A\right)$ denote the variables occurring in $A$.
$\mathop{\exists}\limits^{-}{}_{A}F$ denotes $\exists X_1,\ldots \exists X_n F$ where vars(A) $\backslash$ vars(F) = $\{X_1,\ldots,X_n\}$ \cite{DBLP:conf/ppdp/KoninckSD07,DBLP:conf/iclp/DuckSBH04}
\end{definition}

Table \ref{table:omegat} shows the transitions of $\omega_{t}$ (from \cite{chrbook}).
The transitions are presented below:
\setlength{\textfloatsep}{0cm} 
\setlength{\abovecaptionskip}{0cm}
\begin{table}[ht]
  \begin{tabular}{l}
    \hline
    1. \textbf{Solve}
			: $\langle\{c\} \uplus G, S,B,T\rangle_{n} \mapsto_{solve} \langle G,S,B',T\rangle_{n}$ 
			\\given that $c$ is a built-in constraint and $\mycal{CT} \models \forall((c\wedge B \leftrightarrow B'))$\\
    \hline
		2. \textbf{Introduce}
		: $\langle\{c\} \uplus G, S,B,T>_{n} \mapsto_{introduce} \langle G,S \cup \{c\#n\},B,T\rangle_{n+1}$ 
		\\given that $c$ is a \textsf{CHR} constraint.\\
		 \hline
		3. \textbf{Apply}
		: $\langle G, H_{k}\uplus H_{r}\uplus S,B,T\rangle_{n} \mapsto_{apply}$
		\\$\langle C \uplus G , H_{k}\cup S ,$
		{$chr\left(H_{k}\right)=(H'_{k})\wedge chr\left(H_{r}\right)=(H'_{r}) 
		 \wedge g \wedge B$} $,T \cup \{ \langle r,id\left(H_{k}\right) + id\left(H_{r}\right) \rangle \}\rangle_{n}$ \\
		where 
		{there is a renamed rule in $P$ with variables $x'$ with the form}\\
		$r\; @\; H'_{k}\;\backslash \; H'_{r} \Leftrightarrow \; g \; | \; C.$\\
		such that 
		{$ \mycal{CT} \models \exists(B) \wedge \forall (B \implies \exists x'($
		$chr\left(H_{k}\right)=(H'_{k})\wedge chr\left(H_{r}\right)=(H'_{r}) \wedge g))$}
		\\{and $\langle r,id\left(H_{k}\right) + id\left(H_{r}\right) \rangle\notin T$}
		
		 \\\hline
  \end{tabular}
	\caption{Transitions of $\omega_{t}$}
	\label{table:omegat}
	\end{table}
		\\\textbf{Solve}
	\\The transition \emph{solve}, adds the built-in constraint $c$, in the goal store, to the built-in constraint store. The new built-in store is basically $B \wedge c$. The new store is also simplified. 
			\\\\\textbf{Introduce}
	\\ Similarly, \emph{Introduce}, adds a \textsf{CHR} constraint ($c$) to the \textsf{CHR} constraint store. The ID of the constraint is $n$. The next available ID is $n+1$. The new \textsf{CHR} store is $S \cup \{c\#n\}$.
	\\\\\textbf{Apply}
	\\The \emph{Apply} transition executes an available \textsf{CHR} rule \verb+r+. 
	When a rule is used, its variables are renamed apart from the program and the current state \cite{chrbook}. The fresh variant of the rule with variables the new variables $x'$: 	$r\; @\; H'_{k}\;\backslash \; H'_{r} \Leftrightarrow \; g \; | \; C$ is executable under 
	the applicability condition: 	{$ \mycal{CT} \models \exists(B) \wedge \forall (B \implies \exists x'($
		$chr\left(H_{k}\right)=(H'_{k})\wedge chr\left(H_{r}\right)=(H'_{r}) \wedge g))$}.
		It first checks the built-in constraint store $B$ for satisfiablity. The rule is also only applied if the \textsf{CHR} constraint store contains matching constraints through the checks $chr\left(H_{k}\right)=(H'_{k})\wedge chr\left(H_{r}\right)=(H'_{r})$.
		Under the matching, the guard $g$ is also checked. It has to be logically entailed by the built-in constraints $B$ under the constraint theory $\mycal{CT}$ \cite{chrbook}.
	To apply the rule, the propagation history should not contain a tuple representing the same constraints associated with the same rule signifying that this is the first time this rule is applied with those constraint(s).
	The constraints in the body are added to the goal. In addition, the propagation history is updated accordingly.
	Table \ref{tab:tableexamplewt} shows the query $min\left(20\right),min\left(8\right),min\left(1\right)$ executed under $\omega_{t}$ for the program shown in Section \ref{subsec:extractingmin}.
\begin{table}[h]
\begin{tabular}{ l l p{7cm} }
  & $ \langle \{min\left(20\right),min\left(1\right),min\left(8\right)\},\phi \rangle_{1}$  & \(\) \\
	
 $ \xmapsto[\text{introduce}]{}\mathrel{\vphantom{\to}^*} $ & $ 
\langle \{min\left(1\right),min\left(8\right)\},\{min\left(20\right)\#1\} \rangle_{2}$  & The constraint $min\left(20\right)$ is introduced to the constraint store. \\
$ \xmapsto[\text{introduce}]{}\mathrel{\vphantom{\to}^*} $ & $ 
\langle \{min\left(8\right)\},\{min\left(20\right)\#1,min\left(1\right)\#2\} \rangle_{3}$  & The constraint $min\left(1\right)$ is introduced to the store. \\
$ \xmapsto[\text{apply}]{}\mathrel{\vphantom{\to}^*} $ & $ 
\langle \{min\left(8\right)\},\{min\left(1\right)\#2\} \rangle_{3}$  & The rule \verb+remove_min+ is fired removing \verb+min(20)+ from the constraint store. \\
$ \xmapsto[\text{introduce}]{}\mathrel{\vphantom{\to}^*} $ & $ 
\langle \{\},\{min\left(1\right)\#2,min\left(8\right)\#3\} \rangle_{4}$  & The constraint $min\left(8\right)$ is introduced to the store. \\
$ \xmapsto[\text{apply}]{}\mathrel{\vphantom{\to}^*} $ & $ 
\langle \{\},\{min\left(1\right)\#2\} \rangle_{4}$  & The rule \verb+remove_min+ is fired removing \verb+min(8)+ from the constraint store. \\
 \hline
\end{tabular}

 \caption{Query $min\left(20\right),min\left(8\right),min\left(1\right)$ running under $w_t$}
 \label{tab:tableexamplewt}
\end{table}		
\subsection{Refined Operational Semantics $\omega_{r}$}
\label{sec:refined}
The refined operational semantics \cite{refined,chrbook} is adapted in most implementations of \textsf{CHR}. It removes some of the sources of the non-determinism that exists in $w_{t}$. In $w_{t}$ the order in which constraints are processed and the order of rule application in non-deterministic. However, in $w_{r}$, rules are executed in a top-down manner. Thus, in the case where there are two matching rules , $w_{r}$ ensures that the rule that appears on top in the program is executed. Details of how $w_{r}$ works are shown in this section.
Each atomic head constraint is associated with a number (occurrence). Numbering starts from $1$. It follows a top-down approach as well.
For example the previously shown program to find the minimum value is numbered as follows:
\begin{verbatim}
    remove_dup @ min(X)_2\min(X)_1<=>true.
    remove_min @ min(X)_4\min(Y)_3<=>X<Y | true.
\end{verbatim}

\begin{definition}
The active/occurrenced constraint c\#i:j refers to a numbered constraint that should only match with occurrence j of the constraint c inside the program. $i$ is the identifier of the constraint \cite{refined}.
\end{definition}

A state in $w_{r}$ is the tuple $<A,S,B,T>_{n}$. Unlike $w_{t}$, the goal $A$ is a stack instead of a multi-set. $S,B, T$ and $n$ have the same interpretation as an $w_{t}$ state.
In the refined operational semantics, constraints are executed similar to procedure calls. Each constraint added to the store is activated. An active constraint searches for an applicable rule. The rule search is done in a top-down approach. If a rule matches, the newly added constraints (from the body of the applied rule) could in turn fire new rules. Once all rules are fired, execution resumes from the same point.
Constraints in the constraint store are reconsidered/woken if a newly added built-in constraint could affect them (according to the wakeup policy).
An active constraint thus tries to match with all the rules in the program

Table \ref{table:omegaref} shows the transitions of $w_{r}$.
\setlength{\textfloatsep}{0cm} 
\setlength{\abovecaptionskip}{0cm}
\begin{table}[!ht]
  \begin{tabular}{l}
    \hline
    \\
		1. \textbf{Solve+wakeup}
			: $\langle \left[c|A\right],  S_{0}\uplus S_{1},B,T\rangle_{n} 
			\mapsto_{solve+wake}
			 \langle S_{1}++A, S_{0}\uplus S_{1},B',T\rangle_{n}$
			\\given that $c$ is a built-in constraint and $\mycal{CT} \models \forall((c\wedge B \leftrightarrow B'))$\\
			and $wakeup\left(S_{0}\uplus S_{1} ,c,B\right) = S_{1}$ \\\\
    \hline
		\\2. \textbf{Activate}
		$\langle \left[c|A\right],  S,B,T\rangle_{n} 
			\mapsto_{activate}
			\langle \left[c\#n:1|A\right],  {c\#n} \cup S,B,T\rangle_{n+1}
			$
		\\given that $c$ is a \textsf{CHR} constraint.
		\\
		 \\\hline
		 	\\3. \textbf{Reactivate}
		$\langle \left[c\#i|A\right],  S,B,T\rangle_{n} 
			\mapsto_{reactivate}
			\langle \left[c\#i:1|A\right],   S,B,T\rangle_{n}
			$
		\\given that $c$ is a \textsf{CHR} constraint.
		\\
		 \\\hline
		 	\\4. \textbf{Apply}
		$\langle \left[c\#i:j|A\right], H_1 \uplus H_2 \uplus S,B,T\rangle_{n} 
			\mapsto_{apply\;r}$
			\\$\langle C + H + A,  H_1 \cup S,$
			{$chr\left(H_{1}\right)=(H'_{1})\wedge chr\left(H_{2}\right)=(H'_{2}) 
		 \wedge g \wedge B$} $,
			T \cup \{\left(r,id\left(H_1\right)+id\left(H_2\right)\right)\}\rangle_{n}
			$
		\\given that the jth occurrence of $c$ is part of the head of the re-named apart rule with variables $x'$:
		\\{\(r\; @\; H'_1\; \backslash \;H'_2 \; \Leftrightarrow \;g \;|\; C.\)} 
		\\{where $ \mycal{CT} \models \exists(B) \wedge \forall(B \implies \exists x'\left(\left(chr\left(H_1\right)=(H'_1) \wedge chr\left(H_2\right)=(H'_2)\wedge g\right)\right)$}
		\\and $\left(r,id\left(H_1\right)+id\left(H_2\right)\right) \notin T$.
		\\If $c$ occurs in $H'_1$ then $H=\left[c\#i:j\right]$ otherwise $H=\left[\right]$.
		\\
		 \\\hline
		\\5. \textbf{Drop}
		\\
		$\langle \left[c\#i:j|A\right],  S,B,T\rangle_{n} 
			\mapsto_{drop}
		\langle A,  S,B,T\rangle_{n} 
			$
		\\given that $c\#i:j$ is an occurrenced active constraint and $c$ has no occurrence $j$ in the program. 
		\\That could thus imply that all existing occurrences were tried before.
		\\
		 \\\hline
		\\
		 6. \textbf{Default}
		$\langle \left[c\#i:j|A\right],  S,B,T\rangle_{n} 
			\mapsto_{default}
		\langle \left[c\#i:j+1|A\right],  S,B,T\rangle_{n} 
			$
		\\
		in case there is no other applicable transition.
		\\\\
		 \hline
  \end{tabular}
	\caption{Transitions of $\omega_{r}$}
	\label{table:omegaref}
	\end{table}

\subsubsection{Solve+Wake}
\(\)\\This transition introduces a built-in constraint $c$ to the built-in store. In addition, all constraints that could be affected by $c$ ($S{1}$) are woken up by adding them on top of the stack. These constraints are thus re-activated. A constraint where all its terms have become ground will not be thus woken up by the implemented wake-up policy since it is never affected by a new built-in constraint.
$vars\left(S_{0}\right) \subseteq fixed\left(B\right)$ where $fixed\left(B\right)$ represents the variables fixed by $B$. 

\subsubsection{Activate}
\(\)\\This transition introduces a \textsf{CHR} constraint into the constraint store and activates it. The introduced constraint has the occurrence value $1$ as a start. 

\subsubsection{Reactivate}
\(\)\\The reactivate transition considers a constraint that was already added to the store before. It became re-activated and was added to the stack. The transition activates the constraint by associating it with an occurrence value starting with $1$.

\subsubsection{Apply}
\(\)\\This transition applies a \textsf{CHR} rule $r$ if an active constraint matched a constraint in the head of $r$ with the same occurrence number.
If the matched constraint is part of the constraints to be removed, it is also removed from the stack. Otherwise, it is kept in the constraint store and the stack.

\subsubsection{Drop}
\(\)\\This transition removes the active constraint $c\#i:j$ from the stack when there no more occurrences to check. This occurs when the occurrence number of the active constraint does not appear in the program. In other words, the existing ones were tried.

\subsubsection{Default}
\(\)\\This transition proceeds to the next occurrence of the constraint if the currently active one could not be matched with the associated rule. This transition ensures that all occurrences are tried. 

\subsubsection{Running Example}
\(\)\\In the beginning of the section, a \textsf{CHR} program to find the minimum number among a set of numbers was given. The program was written with the numbered constraint occurrences as used by $\omega_{r}$. For the query min(20),min(1),min(8), the steps shown in Table \ref{tab:table1} take place.

\begin{table}[h]
\begin{tabular}{ l l p{5cm} }
  & $ \langle \left[min\left(20\right),min\left(1\right),min\left(8\right)\right],\phi \rangle_{1}$  & \(\) \\
	
 $ \xmapsto[\text{activate}]{}\mathrel{\vphantom{\to}^*} $ & $ 
\langle \left[min\left(20\right)\#1:1,min\left(1\right),min\left(8\right)\right],\{min\left(20\right)\#1:1\} \rangle_{2}$  & The constraint $min\left(20\right)$ is activated. \\
  $ \xmapsto[\text{default}]{}\mathrel{\vphantom{\to}^*} $ & $ 
	\langle \left[min\left(20\right)\#1:2,min\left(1\right),min\left(8\right)\right],\{min\left(20\right)\#1:1\} \rangle_{2}$  & The active constraint did not fire any rule so the occurrence number is incremented. \\
    $ \xmapsto[\text{default}]{}\mathrel{\vphantom{\to}^*} $ & $ 
		\langle \left[min\left(20\right)\#1:2,min\left(1\right),min\left(8\right)\right],\{min\left(20\right)\#1:1\} \rangle_{2}$  & The occurrence number is incremented. \\
      $ \xmapsto[\text{default}]{}\mathrel{\vphantom{\to}^*} $ & $ 
			\langle \left[min\left(20\right)\#1:3,min\left(1\right),min\left(8\right)\right],\{min\left(20\right)\#1:1\} \rangle_{2}$  & The occurrence number is incremented. \\
        $ \xmapsto[\text{default}]{}\mathrel{\vphantom{\to}^*} $ & $ 
				\langle \left[min\left(20\right)\#1:4,min\left(1\right),min\left(8\right)\right],\{min\left(20\right)\#1:1\} \rangle_{2}$  & The occurrence number is incremented. \\
				$ \xmapsto[\text{default}]{}\mathrel{\vphantom{\to}^*} $ & $ 
				\langle \left[min\left(20\right)\#1:5,min\left(1\right),min\left(8\right)\right],\{min\left(20\right)\#1:1\} \rangle_{2}$  & The occurrence number is incremented. \\
				$ \xmapsto[\text{drop}]{}\mathrel{\vphantom{\to}^*} $ & $ 
				\langle \left[min\left(1\right),min\left(8\right)\right],\{min\left(20\right)\#1:1\} \rangle_{2}$  & The occurrence $5$ does not appear in the program. Thus, $min\left(20\right)\#1:5$ is dropped. \\
				$ \xmapsto[\text{activate}]{}\mathrel{\vphantom{\to}^*} $ & $ 
				\langle \left[min\left(1\right)\#2:1,min\left(8\right)\right],\{min\left(20\right)\#1,min\left(1\right)\#2\} \rangle_{3}$  & The occurrence $5$ does not appear in the program. Thus, $min\left(20\right)\#1:5$ is dropped. \\
					$ \xmapsto[\text{default}]{}\mathrel{\vphantom{\to}^*} $ & $ 
				\langle \left[min\left(1\right)\#2:2,min\left(8\right)\right],\{min\left(20\right)\#1,min\left(1\right)\#2\} \rangle_{3}$  & The occurrence number is incremented. \\
 $ \xmapsto[\text{default}]{}\mathrel{\vphantom{\to}^*} $ & $ 
	\langle \left[min\left(1\right)\#2:3,min\left(8\right)\right],\{min\left(20\right)\#1,min\left(1\right)\#2\} \rangle_{3}$  & The occurrence number is incremented. \\
 $ \xmapsto[\text{default}]{}\mathrel{\vphantom{\to}^*} $ & $ 
	\langle \left[min\left(1\right)\#2:4,min\left(8\right) \right],\{min\left(20\right)\#1,min\left(1\right)\#2\} \rangle_{3}$  & The occurrence number is incremented. \\
	 $ \xmapsto[\text{apply}]{}\mathrel{\vphantom{\to}^*} $ & $ 
	\langle \left[min\left(1\right)\#2:4,min\left(8\right)\right],\{min\left(1\right)\#2\} \rangle_{3}$  & Rule \verb+remove_min+ is applied \\
	$ \xmapsto[\text{default}]{}\mathrel{\vphantom{\to}^*} $ & $ 
	\langle \left[min\left(1\right)\#2:5,min\left(8\right)\right],\{min\left(1\right)\#2\} \rangle_{3}$  & The occurrence number is incremented. \\
		$ \xmapsto[\text{drop}]{}\mathrel{\vphantom{\to}^*} $ & $ 
	\langle \left[min\left(8\right)\right],\{min\left(1\right)\#2\} \rangle_{3}$  & $min\left(1\right)\#1:5$ is dropped since the occurrence $5$ is never used in the program. \\
	$ \xmapsto[\text{activate}]{}\mathrel{\vphantom{\to}^*} $ & $ 
	\langle \left[min\left(8\right)\#3:1\right],\{min\left(1\right)\#2,min\left(8\right)\#3\} \rangle_{4}$  & $min(8)$ is activated. \\
	$ \xmapsto[\text{default}]{}\mathrel{\vphantom{\to}^*} $ & $ 
	\langle \left[min\left(8\right)\#3:2\right],\{min\left(1\right)\#2,min\left(8\right)\#3\} \rangle_{4}$  & The occurrence number is incremented. \\
	$ \xmapsto[\text{default}]{}\mathrel{\vphantom{\to}^*} $ & $ 
	\langle \left[min\left(8\right)\#3:3\right],\{min\left(1\right)\#2,min\left(8\right)\#3\} \rangle_{4}$  & The occurrence number is incremented. \\
		$ \xmapsto[\text{apply}]{}\mathrel{\vphantom{\to}^*} $ & $ 
	\langle \left[\right],\{min\left(1\right)\#2\} \rangle_{4}$  & Rule \verb+remove_min+ is applied. \\
 \hline
\end{tabular}

 \caption{Query $min\left(20\right),min\left(8\right),min\left(1\right)$ running under $w_r$}
 \label{tab:table1}
\end{table}

\section{$CHR^{vis}$: An Animation Extension for CHR}
\label{sec:chrvis}

The proposed extension aims at embedding visualization and animation features into \textsf{CHR} programs. 
The basic idea is that some constraints, the interesting ones, are annotated by visual objects. Thus on adding (removing) such constraints to (from) the constraint store, the corresponding graphical object is added (removed) to (from) the graphical store. 
These constraints are thus treated as interesting events.
Interesting constraints are those constraints that directly represent/affect the basic data structure used along the program. Visualizing such constraints thus leads to a visualization of the execution of the corresponding program. In addition, changes in the constraint store affects the data structure and its visualization. This results in an animation of the execution.
For example, in a program to encode the ``Sudoku'' game, the interesting constraints would be those representing the different cells in the board and their values. Another example is a program solving the $n$-queens problem. The program aims at placing $n$ queens on an $n$ $\times$ $n$ board so that the queens cannot attack each other horizontally, vertically or diagonally.
In this case, the interesting constraints are those reflecting the board and the current/possible locations for each of the $n$ queens \cite{DBLP:journals/corr/SharafAF14,DBLP:conf/lopstr/SharafAF14}.

The approach aims at introducing a generic animation platform independent of the implemented algorithm. This is achieved through two features. First, annotation rules are used. The idea of using interesting events for animating programs was introduced before in Balsa \cite{Brown:1984:SAA:800031.808596} and Zeus \cite{DBLP:conf/vl/Brown91}. Both system use the notion of interesting events. However, users need to know many details to be able to use them.
However, $CHR^{vis}$ eliminated the need of the user knowing any details about the animation.
The second feature is outsourcing the animation into an existing visual tool. For proof of concept, Jawaa \cite{DBLP:conf/sigcse/PiersonR98}, was used. Jawaa provides its users with a wide range of basic structures such as a circle, rectangle, line, textual node , ... etc. Users can also apply actions on Jawaa objects such as movement, changing a color , ... etc.
Users are able to write their own $CHR^{vis}$ programs with the syntax discussed later in this section.
However, users are also provided with an interface (as shown in Figure \ref{fig:ann}) that allows them to specify every interesting event/constraint.
In that case, the programs are automatically generated.
\begin{figure}[!ht]
\centering \includegraphics[width=100mm]{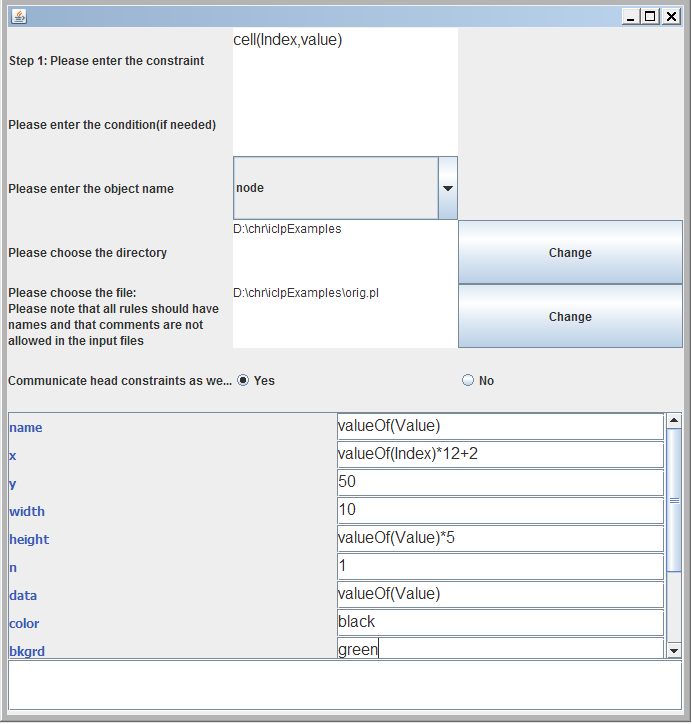}
\caption{Annotating the cell/2 constraint}
\label{fig:ann}
\end{figure}
They are then able to choose the visual object/action (from the list of Jawaa objects/actions) to link the constraint to. Once they make a choice, the panel is populated with the corresponding parameters. Parameters represent the visual properties of the object such as: color, x-coordinate, ... etc. Users have to specify a value for each parameter. A value could be:
\begin{enumerate}
\item a constant value e.g. \verb+100+, \verb+blue+, ... etc.
\item the function \verb+valueOf/1+. valueOf(X) outputs the value of the argument \verb+X+ such that \verb+X+ is one of the arguments of the interesting constraint.
\item the function \verb+prologValue/1+. prologValue(Exp) outputs the value of the argument ``\verb+X+'' computed through the mathematical expression \verb+Exp+.
\item {The keyword} \verb+random+ {that generates a random number}.
\end{enumerate}

\subsection{Extended Programs}
This section introduces the syntax of the \textsf{CHR} programs that are able to produce animations on execution.
In addition to the basic constructs of a \textsf{CHR} program, the extended version needs to specify:
\begin{enumerate}
\item The graphical objects to be used throughout the programs.
\item The interesting constraints and their association with graphical objects.
\item {Whether the head constraints are to be communicated to the tracer.}
\\{In case, the head constraints are communicated, every time a constraint is removed from the store, its} \\{associated visual object is also removed. In some of the cases, this will not be necessary as seen in the below} \\{examples.}
\end{enumerate}

\subsubsection{Syntax of $CHR^{vis}$}
\(\)\\The annotation rules that associate \textsf{CHR} constraint(s) with objects have the following format:
\[
g\;optional\_rule\_name\;H_{vis}\;\Rightarrow\;condition\;|\;graphical\_obj\_name\left(par_{1},par_{2},\ldots,par_{n}\right). 
\]
\(H_{vis}\) contains either one interesting constraint or a group of interesting constraints that are associated with a graphical object. Similar to normal \textsf{CHR} rules, graphical annotation rules could have a pre-condition that has to be satisfied for the rule to be applied. The literal $g$ is added at the beginning of the rule to differentiate between \textsf{CHR} and annotation rules.
A $CHR^{vis}$ program thus has two types of rules. There are the normal \textsf{CHR} rules and the annotation rules responsible for associating \textsf{CHR} constraint(s) with graphical object(s).
Moreover there are meta-annotation rules that associate \textsf{CHR} rules with graphical object(s).
In this case, instead of associating \textsf{CHR} constraint(s) with visual object(s), the association is for a \textsf{CHR} rule. In other words, once such rule is executed the associated visual objects are produced. The association is thus done with the execution of the rule rather than the generation of a new \textsf{CHR} constraint. The rule annotation is done through associating a rule with an auxiliary constraint. The auxiliary constraint has a normal constraint annotation rule with the required visual object.
Such meta-annotation rule has the following format:
\[
g\;optional\_rule\_name\;chr\_rule\_name\;\Rightarrow\;condition\;| \; aux\_constraint\left(par_{1_{aux}},\ldots, par_{m_{aux}}\right). 
\]
\[
g\;aux\_constraint\left(par_{1_{aux}},\ldots, par_{m_{aux}}\right) \;\Rightarrow\ \;graphical\_obj\_name\left(par_{1},par_{2},\ldots,par_{n}\right).
\]

{In addition, a rule for} \verb+comm_head/1+ {has to be added in the} $CHR^{vis}$ {program.}
\\\verb+(comm_head(T) ==> T=true.)+ {means that head constraints are to be communicated to the tracer.}
\\\verb+(comm_head(T) ==> T=false.)+ {means that the removed head constraints should not affect the visualization.}
\subsubsection{Examples}
\(\)\\The program provided in Section \ref{sec:sorting} aims at sorting an array of numbers. In order to animate the execution, we need to visualize the elements of the array. Changes of the elements lead to a change in the visualization and thus animating the algorithm.
The interesting constraint in this case is the \verb+cell+ constraint. As shown in Figure \ref{fig:ann}, it was associated with a rectangular node whose height is a factor of the value of the element. The x-coordinate is a factor of the index. That way, the location and size of a node represent an element of the array.

\fvset{commandchars=\\\{\}}
The new $CHR^{vis}$ program is:
\begin{Verbatim}
:-chr_constraint cell/2.
:-chr_constraint comm_head/1.

{comm_head(T) ==> T=true.}
sort_rule @ cell(In1,V1), cell(In2,V2) <=> In1<In2,V1>V2 | 
                                       cell(In2,V1), cell(In1,V2).
g ann_rule_cell cell(Index,Value) ==> node(valueOf(Value),valueOf(Index)*12+2,
                        50,10,valueOf(Value)*5 ,1,valueOf(Value),
                        black, green, black, RECT).
\end{Verbatim}
Figure \ref{fig:sorting1} shows the result of running the query \verb+cell(0,7),cell(1,6),cell(2,4)+.
As shown from the taken steps, each number added to the array and thus to the constraint store adds a corresponding rectangular node.
Once \verb+cell(0,7)+ and \verb+cell(1,6)+ are added to the constraint store, the rule \verb+sort_rule+ is applicable. Thus, the two constraints are removed from the store. The rule adds \verb+cell(1,7)+ and \verb+cell(0,6)+ to the constraint store.
\begin{figure}[!ht]
\centering
  \subfloat[adding cell(0,7), cell(1,6) to the store]{\hbox{\hspace{2cm}\includegraphics[width=10mm]{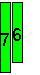}}}
  \subfloat[removing cell(0,7), cell(1,6) from the store]{\hbox{\hspace{2cm}\includegraphics[width=10mm]{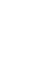}}}
	\hbox{\hspace{2cm}\subfloat[adding cell(1,7), cell(0,6) to the store]{\includegraphics[width=10mm]{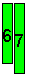}}}
   \hbox{\hspace{2cm}\subfloat[adding cell(2,4) to the store]{\includegraphics[width=10mm]{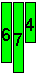}}}
   \hbox{\hspace{2cm}\subfloat[removing cell(0,6) and cell(2,4)  to the store]{\includegraphics[width=10mm]{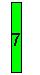}}}
  \hbox{\hspace{2cm}\subfloat[adding cell(2,6)]{\includegraphics[width=10mm]{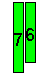}}}
  \hbox{\hspace{2cm}\subfloat[removing cell(1,7),cell(2,6)]{\includegraphics[width=10mm]{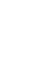}}}
	\hbox{\hspace{2cm}\subfloat[adding cell(2,7),cell(1,6) and cell(0,4)]{\includegraphics[width=10mm]{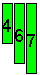}}}
  \caption{Sorting an array of numbers.}
  \label{fig:sorting1}
\end{figure}
Afterwards, \verb+cell(2,4)+ is added to the store. At this point \verb+cell(0,6)+ and \verb+cell(2,4)+ activate \verb+sort_rule+ and are removed from the constraint store. The rule first adds \verb+cell(2,6)+ to the store. At this point \verb+cell(1,7)+ and \verb+cell(2,6)+ activate \verb+sort_rule+ again. Thus they are both removed from the store. The constraints \verb+cell(2,7), cell(1,6)+ are added. Afterwards, the last constraint \verb+cell(0,4)+ is added to the store.
As seen from Figure \ref{fig:sorting1}, using annotations for constraints has helped animate the execution of the sorting algorithm. However, in some of the steps, it might not have been clear which two numbers are being swapped.
In that case it would be useful to use an annotation for the rule \verb+sort_rule+ instead of only annotating the constraint \verb+list+. The resulting program looks as follows:
\begin{Verbatim}
:-chr_constraint cell/2.
:-chr_constraint comm_head/1.

{comm_head(T) ==> T=false.}
sort_rule @ cell(In1,V1), cell(In2,V2) <=> In1<In2,V1>V2 | 
                    cell(In2,V1), cell(In1,V2), swap(In1,V1,In2,V2).
g ann_rule_cell cell(Index,Value) ==> node(nodevalueOf(Value),valueOf(Index)*12+2,50,10, 
                          valueOf(Value)*5 , 1, valueOf(Value), black, 
                          green, black, RECT).
g swap(In1,V1,In2,V2) ==> changeParam(nodevalueOf(V1),bkgrd,pink)
g swap(In1,V1,In2,V2) ==> changeParam(nodevalueOf(V2),bkgrd,pink)
g swap(In1,V1,In2,V2) ==> moveRelative(nodevalueOf(V1),
                          (valueOf(I2)-valueOf(I1))*12,0)
g swap(In1,V1,In2,V2) ==> moveRelative(nodevalueOf(V2),
                          (valueOf(I2)-valueOf(I1))*(-12),0)
g swap(In1,V1,In2,V2) ==> changeParam(nodevalueOf(V1),bkgrd,green)
g swap(In1,V1,In2,V2) ==> changeParam(nodevalueOf(V2),bkgrd,green)
													
g 	sort_rule ==> 	swap(In1,V1,In2,V2).											
\end{Verbatim}
The annotations make sure that once two numbers are swapped, they are first marked with a different color (pink in this case).
The two rectangular bars are then moved. The bar on the left is moved to the right. The bar on the right is moved to the left (negative displacement). The space between the start of one node and the start of the next node is 12 pixels. Thus the displacement is calculated as the difference between the two indeces multiplied by 12. After the swap is done, the two bars are colored back into green.
The result of executing the query: \verb+cell(0,7),cell(1,6),cell(2,4)+ is shown in Figure \ref{fig:sorting2}.
\begin{figure}[!ht]
\centering
	\subfloat[after adding cell(0,7), cell(1,6) to the store, they are marked to be swapped]{\hbox{\hspace{2cm}\includegraphics[width=10mm]{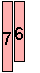}}}
  \subfloat[swapping 7 and 6]{\hbox{\hspace{2cm}\includegraphics[width=10mm]{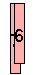}}}
  \subfloat[7 and 6 are swapped]{\hbox{\hspace{2cm}\includegraphics[width=10mm]{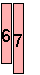}}}
	\hbox{\hspace{2cm}\subfloat[cell(2,4) is added]{\includegraphics[width=10mm]{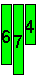}}}
   \hbox{\hspace{2cm}\subfloat[6 and 4 are marked to be swapped]{\includegraphics[width=10mm]{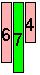}}}
   \hbox{\hspace{2cm}\subfloat[swapping 4 and 6]{\includegraphics[width=10mm]{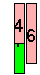}}}
  \hbox{\hspace{2cm}\subfloat[7 and 6 are marked to be swapped]{\includegraphics[width=10mm]{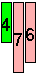}}}
  \hbox{\hspace{2cm}\subfloat[swapping 7 and 6]{\includegraphics[width=10mm]{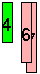}}}
	\hbox{\hspace{2cm}\subfloat[final sorted list]{\includegraphics[width=10mm]{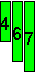}}}
  \caption{Sorting an array of numbers through a rule annotation.}
  \label{fig:sorting2}
\end{figure}

\subsection{Animation Formalization}
\label{sec:form}
The rest of the section offers a formalization of the animation to be able to run $CHR^{vis}$ programs and reason about their correctness. The basic idea is introducing a new ``graphical" store.
$CHR^{vis}$ adds, besides the classical constraint store of \textsf{CHR}, a new store called the graphical store. As implied by the name, the graphical store contains graphical/visual objects. Such objects are the visual mappings of the interesting constraints. 
Over the course of the program execution, and as a result of applying the different rules, the constraint store as well as the graphical store change. As introduced before, the change of the visual objects leads to an animation of the program.

\begin{definition}
In $CHR^{vis}$, a state is represented by a tuple $\langle G,S,Gr,B,T,H\_ann \rangle_{n}$. $G$, $S$, $B$, $T$, and $n$ have the same meanings as in a normal \textsf{CHR} state (goal store, \textsf{CHR} constraint store, built-in store, propagation history and the next available identification number). $Gr$ is a store of graphical objects. 
$H\_ann$ is the history of the applications of the visual annotation rules.
Each element in $H\_ann$ has the following format: $\langle rule\_name, Head\_ids, Object\_ids\rangle$ where
\begin{itemize}
    \item $rule\_name$ represents the name of the fired annotation rule.
    \item $Head\_ids$ contain the ids of the head constraints that fired the annotation rule.
    \item $Object\_ids$ are the ids of the graphical objects added to the graphical store through firing $rule\_name$ using $Head\_ids$.
\end{itemize}

\end{definition}

\begin{definition}
For a sequence $Sq=\left(c_1\#id_1,\ldots,c_n\#id_n\right)$, the function $get\_constraints\left(Sq\right)=\left(c_1 \ldots,c_n\right)$.
\end{definition}

\begin{definition}
\label{def:seqequiv}
Two sequences $A$ and $B$ are equivalent: $A \;\doteq\;B$ if 
\begin{enumerate}
\item For every $X$, if $X$ exists $N$ times in $A$ such that $N>0$, then $X$ exists $N$ times in $B$.
\item For every $Y$, if $Y$ exists $N$ times in $B$ such that $N>0$, then $Y$ exists $N$ times in $A$.
\end{enumerate}
\end{definition}

\begin{definition}
\(\)\\
{{The function} $output\_graphical\_object \left(c\left(Arg_0,\ldots,Arg_n\right), \{Arg'_0,\ldots,Arg'_n\}, 
output\left(Object,OArg_0,\ldots,
OArg_k \right) \right) \\=graphical\_object\left(Actual_0,\ldots,Actual_k\right)$ such that:}
\begin{itemize}
	\item $graphical\_object = Object$.
	\item Each parameter $Actual_n=get\_actual\left(OArg_n\right)$ such that
		\begin{itemize}
			\item if $OArg_n$ is a constant value then $get\_actual\left(OArg_n\right)=OArg_n$.
			\item if {$OArg_n=valueOf\left(Arg_m\right)$ then $get\_actual\left(OArg_n\right)=(Arg'_n)$.}
			\item if $OArg_n=prologValue\left(Expr\right)$ then $get\_actual\left(OArg_n\right)=X$ where $Expr$ is evaluated in SWI-Prolog and binds the variable $X$ to a value.
			\item if  $OArg_n=random$ , then $get\_actual\left(OArg_n\right)$ is a randomly computed number.
		\end{itemize}
\end{itemize}
\end{definition}

\begin{definition}
\(\)\\The function $update\_graphical\_store\left(\{Obj_1\#id_1,\ldots,Obj_i\#id_i\}, graphical\_action\left(Actual_0,\ldots,Actual_k\right)\right)=$
\[ Each\;Obj_i' =
  \begin{cases}
  Obj_i'    & \quad \text{if } Obj_i \text{ is affected by }graphical\_action\left(Actual_0,\ldots,Actual_k\right)\\
     Obj_i   & \quad \text{if } Obj_i \text{ is not affected by }graphical\_action\left(Actual_0,\ldots,Actual_k\right)\\
  \end{cases}
\]
Any $Obj_i'$ could have a different graphical aspect such as its color, x-coordinate, ... etc.
\end{definition}

\begin{definition}
\(\)\\{The function} ${generate\_new\_ann\_history\left(Graph\_obj,Obj\_id,rule\_name,Head\_id,H\_ann\right)=H'\_ann}$ such that:
 in the case where ${\langle rule\_name,Head\_id,Objects\_ids \rangle \in H\_ann}$, 
    \\${H'\_ann=H\_ann - \langle rule\_name,Head\_id,Objects\_ids \rangle \cup \langle rule\_name,Head\_id,Objects\_ids \cup Obj\_id \rangle }$,
\end{definition}

\begin{definition}
\(\)\\{The function} ${remove\_gr\_obj\left(G\_store,rem\_head\_id,H\_ann\right)=G'\_store}$ such that:
 in the case where 
 \\${\langle rule\_name,head\_ids,Objects\_ids \rangle \in H\_ann} \wedge rem\_head\_id \subseteq head\_ids$, 
    \\${G\_store=G\_store - \cup_{i} \left(Obj_{i} \;where \;Obj_{i} \in Objects\_ids\right)}
    $.
   
\end{definition}

\begin{definition}
\label{defcont}
\(\)\\The function $contains\left(H\_ann,\langle rule, Head_{ids}\rangle\right)$ is:
\begin{itemize}
    \item $true$ in the case where $H\_ann$ contains a tuple of the form $\langle rule, Head_{ids}, Objects \rangle$.
    \item $false$ in the case where $H\_ann$ does not contain a tuple of the form $\langle rule, Head_{ids}, Objects \rangle$.
\end{itemize}
\end{definition}
\begin{longtable}{l} 
    \hline\\
			1. \textbf{Solve+wakeup}
			: $\langle \left[c|A\right],  S_{0}\uplus S_{1},Gr,B,T,H\_ann\rangle_{n} 
			\mapsto_{solve+wake}
			 \langle S_{1}++A, S_{0}\uplus S_{1},Gr,B',T,H\_ann\rangle_{n}$
			\\given that $c$ is a built-in constraint and $\mycal{CT} \models \forall\left(\left(c\wedge B \leftrightarrow B' \right)\right)$\\
			and $wakeup\left(S_{0}\uplus S_{1} ,c,B\right) = S_{1}$ \\
    \\\hline\\
			2. \textbf{Activate}
		$\langle \left[c|A\right],  S,Gr,B,T,H\_ann\rangle_{n} 
			\mapsto_{activate}
			\langle \left[c\#n:1|A\right],  \{c\#n\} \cup S,Gr,B,T,H\_ann\rangle_{n+1}
			$
		\\given that $c$ is a \textsf{CHR} constraint.\\
		 \\\hline\\
		 	3. \textbf{Reactivate}
		$\langle \left[c\#i|A\right],  S,Gr,B,T,H\_ann\rangle_{n} 
			\mapsto_{reactivate}
			\langle \left[c\#i:1|A\right],   S,Gr,B,T,H\_ann\rangle_{n}
			$
		\\given that $c$ is a \textsf{CHR} constraint.
		\\
		 \\\hline\\
			4. \textbf{Draw}
		: $\langle \left[{\langle Obj\# 
		\langle r, id\left(H\right),Obj\_ids\rangle}|A\right], S,Gr,B,T,H\_ann\rangle_{n} \mapsto_{draw} \langle A,S ,Gr\cup {\{Obj\#n\}},B,T,H\_ann'\rangle_{n+1}$ \\given that {$Obj$} is a graphical object: $graphical\_object\left(Actual_0,\ldots,Actual_k\right)$. 
		\\and${H\_ann'=generate\_new\_ann\_history\left(Obj,n,r,id\left(H\right),H\_ann\right)}$
		\\The actual parameters of $graphical\_object$ are used to visually render the object.\\
		 \\\hline\\
		 	5. \textbf{Update Store}
		: $\langle \left[{\langle Obj\# 
		\langle r, id\left(H\right),Obj\_ids\rangle}|A\right], S,Gr,B,T,H\_ann\rangle_{n} \mapsto_{update store} \langle A,S ,Gr',B,T,H\_ann  \rangle_{n}$ \\given that {$Obj$} is a graphical action: $graphical\_action\left(Actual_0,\ldots,Actual_k\right)$. 
		\\$Gr'=update\_graphical\_store\left(Gr, graphical\_action\left(Actual_0,\ldots,Actual_k\right)\right)$
		\\The actual parameters of $graphical\_action$ are used to update the graphical objects.\\
		 \\\hline\\
				6. \textbf{{Apply\_Annotation}}:
		\\$\langle \left[c\#i:j|G\right],H \uplus S,Gr,B,T,H\_ann\rangle_{n} \mapsto_{apply\_annotation}$ \\$\langle {\left[Obj\# 
		\langle r, id\left(H\right),\{\;\}\rangle,
		c\#i:j|G\right],} 
		H \cup S ,Gr ,
		B,$ 
		$T , H\_ann \cup \{ \langle r, id\left(H\right),\{\} \rangle \} \rangle_{n}$ 
		where there is: 
		 \\a renamed, constraint annotation rule {with variables $y'$} of the form:
		$g\; r\; @\; H'\;$ 
		$==> Condition\; |\; Obj'$ 
		\\where $c$ is part of $H'$ and
		\\$\mycal(CT) \models \exists \left(B\right) 
		\wedge \forall \left(B \implies \exists y'\left(chr\left(H\right)=(H') \wedge Condition \wedge output\_graphical\_object \left(H',y', Obj'\right)=Obj\right)  \right)$ \\and $\neg \left(contains\left(H\_ann,\left(r, id\left(H\right)\right)\right)\right)$ 
		
		\\
	\\\hline\\
				7. \textbf{{Apply}}
		: $\langle \left[c\#i:j|G\right], H_{k}\uplus H_{r}\uplus S,Gr,B,T,H\_ann\rangle_{n} \mapsto_{apply}$
		\\$\langle C ++ H ++ G , H_{k}\cup S ,Gr',$
		{$chr\left(H_{k}\right)=(H'_{k})\wedge chr\left(H_{r}\right)=(H'_{r}) 
		 \wedge g \wedge B$}
		\\$T \cup \{ \langle r,id\left(H_{k}\right) + id\left(H_{r}\right) \rangle \},H\_ann\rangle_{n}$ 
		where:
		\\\tabitem there is no applicable constraint annotation rule that involves $c$. 
		\\(i.e. every applicable rule has already been applied). \\In other words, for renamed-apart every annotation rule with variables $y'$:  
		\\$g\; r\; @\; H'\; ==>\; Cond\; |\; Obj'$,
		\\
		{$c$ is part of $H'$}
		{$\mycal(CT) \models \exists \left(B\right) \wedge \forall \left(B \implies \exists y'  (chr\left(H\right)=(H') \wedge Condition)\right)$ }
		\\It is already the case that: 
		$\left(contains\left(H\_ann,\left(r, id\left(H\right)\right)\right)\right)=true$
		
		\\\tabitem 
		There is a renamed rule in $P_{vis}$ with the form\\
		$r\; @\; H'_{k}\;\backslash \; H'_{r} \Leftrightarrow \; g \; | \; C.$\\ with variables $x'$
		and the jth occurrence of $c$ is part of the head of the renamed rule,
		such that 
		\\where \( \mycal{CT} \models \exists \left(B\right) \wedge \forall(B \implies \exists x'\left((chr\left(H_k\right)=(H'_k) \wedge chr\left(H_r\right)=(H'_r)\wedge g)\right)\) 
		\\and \(\langle r,id\left(H_{k}\right) + id\left(H_{r}\right) \rangle\notin T\).
		\\If $c$ occurs in $H'_k$ then $H=\left[c\#i:j\right]$ otherwise $H=\left[\right]$.
		\\If the program communicates the head constraints (i.e. contains \verb+comm_head(T) ==> T=true+) then
		\\${Gr'= remove\_gr\_obj\left(G,id\left(H'_{r}\right),H\_ann\right)}$
		\\\hline\\
8. \textbf{Drop}
		\\
		$\langle \left[c\#i:j|A\right],  S,Gr,B,T,H\_ann\rangle_{n} 
			\mapsto_{drop}
		\langle A,  S,Gr,B,T,H\_ann\rangle_{n} 
			$
		\\given that $c\#i:j$ is an occurrenced active constraint and $c$ has no occurrence $j$ in the program 
		\\and that there is no applicable constraint annotation rule for the constraint $c$.
		\\That could thus imply that all existing ones were tried before.
		\\
		 \\\hline\\
		 8. \textbf{Default}
		$\langle \left[c\#i:j|A\right],  S,Gr,B,T,H\_ann\rangle_{n} 
			\mapsto_{default}
		\langle \left[c\#i:j+1|A\right],  S,Gr,B,T,H\_ann\rangle_{n} 
			$
		\\
		in case there is no other applicable transition.
		\\
		\\\hline
\caption{Transitions of \(\omega_{vis}\)}
	\label{table:omegavis}
\end{longtable}

	Table \ref{table:omegavis} shows the basic transitions of $\omega_{vis}$. To make the transitions easier to follow, table \ref{table:omegavis}  shows the transitions needed to run \textsf{CHR} programs with constraint annotation rules. Annotations of \textsf{CHR} rules are thus discarded from the first set of transitions.
	$\omega_{vis}$ allows for running programs that contain constraint annotations.
The three transitions $apply\_annotation$ and $draw$ are responsible for dealing with the graphical store and its constituents. The transition, $apply\_annotation$, applies a constraint annotation rule. The rest of the transitions, such as solve, introduce and apply, have the same behavior as in $\omega_{r}$. 
These transitions do not affect the graphical store or the application history of the annotation rules.

	The transitions affecting the graphical store are:
	\begin{enumerate}
	\item \textbf{Draw}
	\\The new transition \emph{draw} adds a graphical object ($g$) to the graphical store. Since multiple copies of a graphical object are allowed, each object is associated with a unique identifier.
\item \textbf{Apply\_Annotation}
	\\The \emph{Apply\_Annotation} transition applies a constraint annotation rule (ann\_rule). An annotation rule is applicable if the \textsf{CHR} constraint store contains matching constraints. The condition of the rule has to be implied by the built in store under the matching. The built in constraint store B is also first checked for satisfiability. For the rule to be applied, it should not have appeared in the history of applied annotation rules with the same constraints i.e. it should be the first time the constraint(s) fire this annotation rule. Executing the rule adds to the goal the graphical object in the body of the executed annotation rule. The history of annotation rules is updated accordingly with the name of the rule in addition to the id(s) of the \textsf{CHR} constraint(s) in the head.	
	In fact, this transition has a higher precedence than the transition $apply$. Thus in the case where an annotation rule and a \textsf{CHR} rule are applicable, the annotation rule is triggered first.
	The precedence makes sure that graphical objects are added in the intended order. This ensures producing correct animations.
		\end{enumerate}

\begin{definition}[Built-In Store Equivalence]
\label{builtinequiv}
{
\(\)\\Two built-in constraint stores $B_{1}$ and $B_{2}$ are considered equivalent iff:
\\$\mycal(CT) \models \forall(\exists_{y1}(B_{1}) \leftrightarrow \exists_{y2}(B_{2}))$
where $y_1$ and $y_2$ are the local variables inside $B_1$ and $B_2$ respectively.
The equivalence thus basically ensures that there are no contradictions in the substitutions since local variables are renamed apart in every \textsf{CHR} program. The equivalence check thus ensures the logical equivalence rather than the syntactical equivalence.
}
\end{definition}

	\begin{definition}
\label{def:eqchrvis}
A $CHR^{vis}$ state $St_{vis}=\langle G_{vis},S_{vis},Gr_{vis},B_{vis},T_{vis},T_{visAnn} \rangle_{n_{vis}}$ is equivalent to a \textsf{CHR} state $St=\langle G,S,B,T \rangle_{n}$ if and only if 
\begin{enumerate}
	\item $get\_constraints\left(G_{vis}\right) \doteq get\_constraints\left(G\right)$ according to Definition \ref{def:seqequiv}.
	\item $get\_constraints\left(St_{vis}\right) \doteq get\_constraints\left(S\right)=C$ according to Definition \ref{def:seqequiv}.
	\item {$B_{vis}$ and $B$ are equivalent according to Definition} \ref{builtinequiv}.
	\item $T_{vis} = T$
	\item $n_{vis} \geq n$
\end{enumerate}

The idea is that a $CHR^{vis}$ state basically has an extra graphical store. The correspondence check is effectively done through the \textsf{CHR} constraints since they are the most distinguishing constituents of a state. Thus, the constraint store and the stack should contain the same constraints. The propagation history should be also the same indicating that the same \textsf{CHR} rules have been applied. $n_{vis}$ could, however, have a value higher than $n$. This is due to the fact that graphical objects have identifiers.

{The definition of state equivalence described here follows the properties introduced in} \cite{raiserequiv}. {However, it is stricter.
}
\end{definition}

\begin{theorem}[Completeness]
\label{proof:pr1}
Given a \textsf{CHR} program $P$ (running under $\omega_{r}$) along with its user defined annotations and its corresponding $P_{CHR^{vis}}$ (running under $\omega_{vis}$) program,
for the same query $Q$, every derived state $S_{chr}$: $Q \mapsto_{\omega_{r}}^{*} S_{chr}$ has an equivalent state $S_{chr_{vis}}$: $Q \mapsto_{\omega_{vis}}^{*} S_{chr_{vis}}$.
\begin{proof}
\label{pr:pro1}
\(\)\\Base Case: For a given query $Q$, the initial state in $\omega_{r}$ $S_{chr}= \langle Q,\{\},\{\},\{\} \rangle_1$. The initial state in $\omega_{vis}$ is $S_{chr_{vis}}=\langle Q,\{\},\{\},\{\},\{\},\{\} \rangle_1$.
\footnote{Throughout the different proofs, identifiers are omitted for brevity} According to Definition \ref{def:eqchrvis} $S_{chr}$ and $S_{chr_{vis}}$ are equivalent.\\ 
\\Induction Hypothesis:
Suppose that there are two equivalent derived states $S_{chr}=\langle A,S,B,T \rangle_n$ and $S_{chr_{vis}}=\langle A,S,Gr,B,T,H\_ann \rangle_m$ such that $Q \mapsto_{\omega_{r}}^{i} S_{chr}$ and $Q \mapsto_{\omega_{vis}}^{j} S_{chr_{vis}}$.\\ 
\\\emph{Induction Step}:
According to the induction hypothesis, $S_{chr}$ and $S_{chr_{vis}}$ are equivalent.
The rest of the proof shows that any transition applicable to $S_{chr}$ in $\omega_{r}$ produces a state that has an equivalent state produced by applying a transition to $S_{chr_{vis}}$ in $\omega_{vis}$.
Thus, no matter how many times the step is repeated, the output states are equivalent.
\begin{itemize}
	\item \textbf{Case 1: (Applying Solve+wakeup):}
	\\In this case, $S_{chr}\mapsto S_{chr}'$ such that:
	\\$S_{chr}: \langle \left[c|A\right],  S_{0}\uplus S_{1},B,T\rangle_{n} 
			\mapsto_{solve+wake}
			 \langle S_{1}++A, S_{0}\uplus S_{1},B',T\rangle_{n}$
			\\Transition \emph{solve+wakeup} is applicable if:
			\begin{enumerate}
				\item $c$ is a built-in constraint
				\item $\mycal{CT} \models \forall((c\wedge B \leftrightarrow B'))$
				\item $wakeup\left(S_{0}\uplus S_{1} ,c,B\right) = S_{1}$
			\end{enumerate}
			$S_{chr_{vis}}(\langle Stack,  S_{chr_{vis}},Gr,B_{vis},T_{vis},T_{ann}\rangle_{m})$ is equivalent to $S_{chr}(\langle \left[c|A\right],  S_{0}\uplus S_{1},B,T\rangle_{n})$. Thus according to Definition \ref{def:eqchrvis},
				 $Stack = \left[c|A\right]$
				$\wedge$ $S_{chr_{vis}}=S_{0}\uplus S_{1}$
				$\wedge$ $B_{vis}=B$
				$\wedge$ $T_{vis}=T$
				$\wedge m \geq n$
			Thus accordingly, the transition $solve+wakeup$ is applicable to $S_{chr_{vis}}$ under $\omega_{vis}$ producing $S_{chr_{vis}}'$:$\langle S_{1}++A, S_{0}\uplus S_{1},Gr,B\wedge c,T,H\_ann\rangle_{m}$. According to Definition \ref{def:eqchrvis}, $S'_{vis}$ is equivalent to $S_{chr}'$\\
			\item \textbf{Case 2: (Applying Activate):}
			\\In this case, $S_{chr}=\langle \left[c|A\right],  S,B,T\rangle_{n}$ where $c$ is a \textsf{CHR} constraint. Thus $S_{chr}\mapsto_{activate} S_{chr}': \langle \left[c\#n:1|A\right],  {c\#n} \cup S,B,T\rangle_{n+1}$.\\	Since $S_{chr_{vis}}(\langle Stack,  S_{chr_{vis}},Gr,B_{vis},T_{vis},T_{ann}\rangle_{m})$ is equivalent to $S_{chr}(\langle \left[c|A\right],  S_{0}\uplus S_{1},B,T\rangle_{n})$. Thus according to Definition \ref{def:eqchrvis}:
				 $Stack = \left[c|A\right]$
				$\wedge$ $S_{chr_{vis}}=S$
				$\wedge$ $B_{vis}=B$
				$\wedge$ $T_{vis}=T$
				$\wedge$ $m \geq n$\\
				Accordingly, $S_{chr_{vis}} \mapsto_{activate} S_{chr_{vis}}':\langle \left[c\#n:1|A\right],  \{c\#n\} \cup S,Gr,B,T,T_{ann}\rangle_{m+1}$ which is equivalent to $S_{chr}'$. (Since $m \geq n$, then $m+1 \geq n+1$).\\
				\item \textbf{Case 3 (Applying Reactivate):}
				\\The transition \emph{reactivate} is applicable if the stack has on top of it an element of the form $c\#i$ where $c$ is a \textsf{CHR} constraint. In this case $S_{chr}
				=\langle \left[c\#i|A\right],  S,B,T\rangle_{n} $. Accordingly, $S_{chr}
			\mapsto_{reactivate}
			S_{chr}':\langle \left[c\#i:1|A\right],   S,B,T\rangle_{n}
			$. Since $S_{chr_{vis}}$ and $S_{chr}$ are equivalent, then $S_{chr_{vis}}$ has the same stack.  $S_{chr_{vis}}=\langle \left[c\#i|A\right],  S,Gr,B,T,T_{ann}\rangle_{m}$ triggers the transition reactivate producing $S_{chr_{vis}}':\langle [c\#i:1|A],  S,Gr,B,T,T_{ann}\rangle_{m}$ which is also equivalent to $S_{chr}'$. 
			{Since $c$ is not associated with an occurrence yet, no annotation rule is applicable at this point.}
		\\
		\item \textbf{Case 4: (Applying the transition Apply)}
		\\The transition \emph{Apply} is triggered under $\omega_{r}$ in the case where $S_{chr}=\langle \left[c\#i:j|A\right], H_1 \uplus H_2 \uplus S,B,T\rangle_{n} $
			such that the jth occurrence of $c$ is part of the head of the re-named apart rule {with variables $x'$}:
		{\(r\; @\; H'_1\; \backslash \;H'_2 \; \Leftrightarrow \;g \;|\; C.\)}
		\\
		such that:
		\\{$\mycal{CT} \models \exists(B) \wedge \forall
		(B \implies \exists x'( chr\left(H_{1}\right)=(H'_{1})\wedge chr\left(H_{2}\right)=(H'_{2})
		\wedge g)))$ }
		and $\langle r,id\left(H_1\right)+id\left(H_2\right) \rangle \notin T$.
		\\
	{Thus in such a case $S_{chr} \mapsto_{apply\;r}$
			$S_{chr}':\langle C ++ H ++ A,  H_1 \cup S,$}
			{$chr\left(H_{1}\right)=(H'_{1})\wedge chr\left(H_{2}\right)=(H'_{2}) \wedge g \wedge$} $B,
			T \cup \{\langle r,id\left(H_1\right)+id\left(H_2\right)\rangle\}\rangle_{n}$
			
		 \[ H =
  \begin{cases}
   \left [c\#i:j\right]      & \quad \text{if } c \text{ occurs in }H'_1\\
    \left[\;\right] & \quad \text{otherwise}\\
  \end{cases}
\]
		\\Due to the fact that $S_{chr}$ and $S_{chr_{vis}}$ are equivalent, in the case where $S_{chr}$ triggers the transition \emph{Apply} under $\omega_{r}$, the same rule is also applicable under $\omega_{vis}$ to $S_{chr_{vis}}$.
		However for $S_{chr_{vis}}$, one of two possibilities could happen:
		\\\tabitem There is no applicable constraint annotation rule:
		\\{This could be due to the fact that any applicable annotation rule was already executed or that there re no applicable} {annotation rules at this point.}
		In this case,the transition $Apply$ is triggered right away under $\omega_{vis}$ producing a state
		\\
		{
		($S_{chr_{vis}}':\langle C ++ H ++ A,  H_1 \cup S,Gr,
			chr\left(H_1\right) = H'_1 \wedge chr\left(H_2\right) = H'_2 \wedge g \wedge B,$
			\\$T \cup \{\langle r,id\left(H_1\right)+id\left(H_2\right)\rangle, H\_ann\}\rangle_{m}$) equivalent to ($S_{chr}'$).
		}
			{The original states are equivalent and the same rule is applied in both cases. We can assume that,}
			{without loss of generality}
			{, in the $chr_{vis}$ program, the rule is renamed using the same variables $x'$ resulting in the same matching. This is because the same matching should happen to be able to apply the same rule using the given constraint stores}.
		\\\tabitem There is an applicable annotation rule:
		\\In this case an annotation rule ($r_{ann}$) for $c$ is applicable such that:
		\\$S_{chr_{vis}}\langle \left[c\#i:j|A\right],H_1 \uplus H_2 \uplus S,Gr,B,T,H\_ann\rangle_{m} \mapsto_{apply\_annotation}$ 
		\\$S_{chr_{vis}}':\langle \left[Obj\#\langle  r,id\left(H\right),\{\} \rangle,c\#i:j|A\right],H_1 \uplus H_2 \uplus S ,Gr , B,$ 
		$T , H\_ann \cup \{\langle r_{ann}, id\left(H\right),\{\;\}\rangle\} \rangle_{m}$ 
		according to the previously mentioned conditions.
		\\At this point either the transition \emph{Draw} or \emph{Update store} is applicable such that:
		\\$S_{chr_{vis}}' \mapsto_{draw\big/update store} S_{chr_{vis}}'':\langle \left[c\#i:j|A\right], H_1 \uplus H_2 \uplus S ,Gr', B,$ 
		$T , H'\_ann \rangle_{m'}$
		\\In case $Obj$ is a graphical object, the transition \emph{Draw} is applied such that:
		$Gr'=Gr\cup\{Obj\#m\}\;\wedge\;m'=m+1 \;\wedge\; H'\_ann = generate\_new\_ann\_history\left(Obj,m,r,id\left(H\right),\cup \{\langle r_{ann}, id\left(H\right),\{\;\}\rangle\}\right)$.
		\\In case, $Obj$ is a graphical action, the transition \emph{Update Store} is applied such that:
		\\$Gr'=update\_graphical\_store\left(Gr,Obj\right)\;m'=m\;H'\_ann=\cup \{\langle r_{ann}, id\left(H\right),\{\;\}\rangle\}$
		\\Since the two transitions, could only change the graphical stores, annotation history and the next available identifier, the equivalence of the states is not affected.
		\\At this point $\omega_{vis}$ fires the transition $Apply$ for the same \textsf{CHR} rule that triggered the same transition under $\omega_{r}$ earlier. The produced state $S_{chr_{vis}}'''$ has the format: 
		\\{$\langle C ++ H ++ A,  H_1 \cup S, Gr',
			chr\left(H_1\right) = H'_1 \wedge chr\left(H_2\right) = H'_2 \wedge B,
			T \cup \{\langle r,id\left(H_1\right)+id\left(H_2\right)\rangle\},$}
			\\{$H'\_ann\rangle_{m'}$.
			}
			{Similarly the same matching (local variable renaming $x'$) has to be applied for the rule to fire}.
			\\Consequently, according to Definition \ref{def:eqchrvis}, the state $S_{chr_{vis}}'''$ is still equivalent to $S_{chr}'$
			\\
			\item \textbf{Case 5: Applying the transition Drop}
			\\In the case where the top of the stack has an occurrenced active constraint $c\#i:j$ such that $c$ has no occurrence $j$ in the program, the transition drop is applied.
			Thus, $S_{chr}:\langle\left[c\#i:j|A\right],  S,B,T\rangle_{n} 
			\mapsto_{drop}
		S_{chr}': \langle A,  S,B,T\rangle_{n} 
			$
			\\Since $S_{chr_{vis}}$ and $S_{chr}$ are equivalent, the stack of both states have to be equivalent. 
		\\Thus $S_{chr_{vis}}=\langle \left[c\#i:j|A\right],  S,Gr,B,T,H\_ann\rangle_{m}$.
			For $\omega_{vis}$ one of two possibilities is applicable:
			\begin{enumerate}
			    \item No annotation rule is applicable.
			    This could be either because $c$ is not associated with any visual annotation rules or because all such rules have been already applied.
			    In this case \\$S_{chr_{vis}}:\langle\left[c\#i:j|A\right],  S,Gr,B,T,H\_ann\rangle_{m} \mapsto_{drop} S_{CHR_{vis}}':\langle A,  S,Gr,B,T,H\_ann\rangle_{m}$
			    \item The second possibility is the existence of an applicable annotation rule: transforming $S_{chr_{vis}}$ to $S_{chr_{vis}}':\langle\left[Obj\#\langle r,id\left(H\right),\{\;\}\rangle,c\#i:j|A\right],  S,Gr,B,T,H'\_ann\rangle_{m}$.
			    At that point either \emph{draw} or \emph{update store} are to be applied transforming $S_{chr_{vis}}'$ to $S_{chr_{vis}}'':\langle \left[c\#i:j|A\right]$\\$,  S,Gr',B,T,H''\_ann\rangle_{m'}$.
			    At that point, the transition drop is applicable converting $S_{chr_{vis}}''$ to $S_{chr_{vis}}''':
			    \\\langle A,  S,Gr',B,T,H''\_ann\rangle_{m'}$. $S_{chr_{vis}}'''$ is equivalent to $S_{chr}'$
			\end{enumerate}
			\(\)\\
			\item \textbf{Case 6: Applying the Default Transition}
		\\If none of the previous cases is applicable, $S_{chr}:\langle \left[c\#i:j|A\right],  S,B,T\rangle_{n} 
			\mapsto_{drop}
		S_{chr}':\langle \left[c\#i:j+1|A\right],  S,B,T\rangle_{n}.
			$
		\\For the equivalent $S_{chr_{vis}}$, one of two possible cases could happen:
		\begin{enumerate}
			\item \textbf{Apply annotation is not applicable:}
			\\In that case, the \emph{Default} transition is directly applied transforming $S_{chr_{vis}} to S_{chr_{vis}}'$ such that \\$\langle \left[c\#i:j|A\right],  S,Gr,B,T,H\_ann\rangle_{m} \mapsto_{drop} \langle A,  S,Gr,B,T,H\_ann\rangle_{m}$.\\The produced state ($S_{chr_{vis}}'$) is equivalent to $S_{chr}'$ as well.
			\item \textbf{Apply annotation is applicable:}
			\\In this case an annotation rule for one of the existing constraints is applicable such that:
		\\$S_{chr_{vis}}\langle \left[c\#i:j|A\right], S,Gr,B,T,H\_ann\rangle_{m} \mapsto_{apply\_annotation}$ 
		\\$S_{chr_{vis}}':\langle \left[Obj\#\langle r,id(H), \{\;\}\rangle ,c\# i:j|A\right], S ,Gr , B,$ 
		$T , H'\_ann \rangle_{m}$ 
		according to the previously mentioned conditions.
		\\At this point either the transition \emph{Draw} or the transition \emph{Update store} is applicable such that:
		\\$S_{chr_{vis}}' \mapsto_{draw} S_{chr_{vis}}'':\langle \left[c\#i:j|A\right], S ,Gr' , B,$ 
		$T , H''\_ann \rangle_{m'}$
		\\$S_{chr_{vis}}''$ is still equivalent to $S_{chr}$.
		\\At the point where the transition \emph{apply\_annotation} is no longer applicable,  the only applicable transition is \emph{Default} transforming $S_{chr_{vis}}''$ to $S_{chr_{vis}}'''$ such that $S_{chr_{vis}}'''=\langle A,  S ,Gr' , B,$ 
		$T , H''\_ann \rangle_{m'}$. According to Definition \ref{def:eqchrvis}, $S_{chr_{vis}}'''$ is equivalent to $S_{chr}'$
		\end{enumerate}
		\(\)\\
		Thus in all cases an equivalent state is produced under $\omega_{vis}$
		\qed
\end{itemize}

\end{proof}
\end{theorem}

\begin{theorem}[Soundness]
\label{pro:proof2}
Given a \textsf{CHR} program $P$ (running under $\omega_{r}$) along with its user defined annotations and its corresponding $P_{CHR^{vis}}$ program (running under $\omega_{vis}$),
for the same query $Q$, every derived state  $S_{chr_{vis}}$: 
$Q$ $\mapsto_{\omega_{vis}}^{*} $ $S_{chr_{vis}}$
has en equivalent state $S_{chr}$: 
$Q$ $\mapsto_{\omega_{r}}^{*} $ $S_{chr}$

\begin{proof}
\label{pr:pro2}
\(\)\\\emph{Base Case}: 
\\For the initial query the two states $Q$, $S_{chr_{vis}}=\langle Q,\{\},\{\} \rangle$ and $S_{chr} = \langle Q,\{\} \rangle$ are equivalent according to Definition \ref{def:eqchrvis}.
\\\\\emph{Induction Hypothesis}:
Suppose that there are two equivalent derived states  $S_{chr_{vis}}=\langle A,S,Gr,B,T,H\_ann \rangle_m$ and $S_{chr}=\langle A,S,B,T \rangle_n$ such that $Q \mapsto_{\omega_{vis}}^{i} S_{chr_{vis}}$ and $Q \mapsto_{\omega_{r}}^{j} S_{chr}$.\\ 
\\\emph{Induction Step}:
\\The proof shows that any transition applicable to $S_{chr_{vis}}$ under $\omega_{vis}$  produces a state $S_{chr}'$ such that under $\omega_{r}$ applying a transition to $S_{chr}$ (which is equivalent to $S_{chr_{vis}}$) produces a state $S_{chr}'$ that is equivalent to $S_{chr}$.\\
The different cases are enumerated below:
\begin{enumerate}
\item \textbf{Case 1: Applying solve+wakeup to $\mathbf{S_{chr_{vis}}}$:}
\\Under $\omega_{vis}$, solve+wakeup is applicable in the case where the stack has the form $\left[c|A\right]$ such that c is a built-in constraint and $\mycal{CT} \models \forall(\left(c\wedge B \leftrightarrow B'\right))$\\
			and $wakeup\left(S_{0}\uplus S_{1} ,c,B\right) = S_{1}$ such that 
\\$S_{chr_{vis}} 
			\mapsto_{solve+wake}
			 S_{chr_{vis}}':\langle S_{1}++A, S_{0}\uplus S_{1},Gr,B',T,H\_ann\rangle_{m}$.
Since $S_{chr_{vis}}$ and $S_{chr}$ are equivalent, $S_{chr}$ has an equivalent stack and built-in store according to Definition \ref{def:eqchrvis}. Thus the corresponding transition \emph{solve+wakeup} is applicable to $S_{chr}$ under $\omega_{r}$ producing a state $S_{chr}'$ such that: $S_{chr}'=\langle S_{1}++A, S_{0}\uplus S_{1},B',T\rangle_{n}$.
According to Definition \ref{def:eqchrvis}, the two states $S_{chr_{vis}}'$ and $S_{chr}'$ are equivalent.
\\
\item \textbf{Case 2: Applying Activate:}
\\Such a transition is applicable to $S_{chr_{vis}}$ under $\omega_{vis}$ in the case where the top of the stack of $S_{chr_{vis}}$ contains a \textsf{CHR} constraint $c$. In this case:
\\$S_{chr_{vis}}:\langle \left[c|A\right],  S,Gr,B,T,H\_ann\rangle_{m} 
			\mapsto_{activate}
			S_{chr_{vis}}': \langle \left[c\#n:1|A\right],  \{c\#n\} \cup S,Gr,B,T,H\_ann\rangle_{m+1}
			$
		\\given that $c$ is a \textsf{CHR} constraint.\\
The equivalent state $S_{chr}$ has the same stack triggering the transition \emph{Activate} under $\omega_{r}$ producing a state $S_{chr}':\langle \left[c\#n:1|A\right],  \{c\#n\} \cup S,Gr,B,T,H\_ann\rangle_{n+1}$ which is also equivalent to $S_{chr_{vis}}'$
\\
\item \textbf{Case 3: Applying Reactivate:}
\\In this case,
$S_{chr_{vis}} 
			\mapsto_{reactivate}
			S_{chr_{vis}}' \langle \left[c\#i:1|A\right],   S,Gr,B,T,H\_ann\rangle_{m}
			$
\\such that $S_{chr_{vis}}= \langle \left[c\#i|A\right],  S,Gr,B,T,H\_ann\rangle_{m}$
and $c$ is a \textsf{CHR} constraint.
		\\
The equivalent state $S_{chr}$ has an equivalent stack triggering the transition \emph{reactivate} under $\omega_{r}$. The transition application produces $S_{chr}':\langle \left[c\#i:1|A\right],   S,B,T\rangle_{n}$ which is also equivalent to $S_{chr_{vis}}'$.\\
\item According to Definition \ref{def:eqchrvis} and since $S_{chr_{vis}}$ is equivalent to $S_{chr}$, they both have the same stack. The transition \textbf{Draw} is only applicable if the top of the stack contains a graphical object. Since the stack of $S_{chr}$ never contains graphical objects and since it is equivalent to $S_{chr_{vis}}$, the stack of $S_{chr_{vis}}$ at this point does not contain graphical objects as well.
Thus, in this case, the transition \textbf{\emph{Draw}} is not applicable to $S_{chr_{vis}}$ under $\omega_{vis}$.
\\
\item  Similarly, according to Definition \ref{def:eqchrvis} and since $S_{chr_{vis}}$ is equivalent to $S_{chr}$, the stack of $S_{chr_{vis}}$ at this point does not contain graphical actions since both states should have the same stack. The transition \textbf{Update store} is only applicable if the top of the stack contains a graphical action. 
Thus, in this case, the transition \textbf{\emph{Update store}} is not applicable to $S_{chr_{vis}}$ under $\omega_{vis}$.
\\
\item \textbf{Case 4: \emph{Apply Annotation Rule} Transition}
\\The transition \emph{Apply Annotation} is triggered when the {stack} has on top a constraint associated with an annotation rule. The constraint store should contain constraints matching the head of the annotation rule such that this rule was not fired with those constraint(s) before and the pre-condition of the annotation rule is satisfied.  {Thus, the rule could be associated with more than one constraint including} {the one on top of the stack.} \\{The constraint store should however, contain matching constraints for the rest of the constraints in the head of} \\{the annotation rule.}
\\$S_{chr_{vis}} \mapsto_{apply\_annotation}$ $S_{chr_{vis}}':\langle \left[Obj\#\langle r, id\left(H\right), \{\;\}  \rangle|A\right], H \cup S ,Gr ,$
		$B, T , H\_ann \cup \{\langle r, id\left(H\right),\{\;\}\rangle\} \rangle_{m}$ \\ 
		such that $\neg contains\left(r,id\left(H\right)\right)$. The renamed annotation rule with variables $x'$ is :
		\\$g\; r\; @\; H'\;$ 
		$==> Condition\; |\; Obj'$ 
		\\
		{$\mycal(CT) \models \exists \left(B\right) \wedge \forall (B \implies \exists x' ((chr\left(H\right)=(H') \wedge Cond \wedge$} $output\_graphical\_object\left(H',x',Obj'\right)=Obj)))$
		\\Either the transition \emph{Draw} or \emph{Update store} is applicable to $S'_{chr_{vis}}$. The output is $S''_{chr_{vis}}:\langle A, S ,Gr' ,$
		$T , H'\_ann  \rangle_{m'}$.
	In case, $Obj$ is a graphical object, then $H'_ann=generate\_new\_ann_history\left(H\_ann \cup \{\langle r, id\left(H\right),\{\;\}\rangle\}\right)\wedge Gr'=Gr\cup\{Obj\#m\}\wedge m'=m+1$.
	In case, $Obj$ is a graphical action, then $Gr'=update\_graphical\_store\left(Gr,Obj\right) \wedge Gr'=Gr \wedge m'=m$
		Any transition applicable to $S_{chr_{vis}}''$ at this stage is covered through the rest of the cases. Thus the application of the transition $apply\_annotation$ is considered as not to affect the equivalence of the output state with $S_{chr}$.
		\\
\item \textbf{Case 5: the Apply transition}
\\
In the case where a \textsf{CHR} rule is applicable to $S_{chr_{vis}}$, the transition 
\emph{Apply} is triggered under $\omega_{vis}$. A \textsf{CHR} rule $r$ is applicable in the case where a renamed version of the rule $r$ with {variables $x'$}:
($r\; @\; H'_{k}\;\backslash \; H'_{r} \Leftrightarrow \; g \; | \; C.$) where $\langle r,id\left(H_{k}\right) + id\left(H_{r}\right) \rangle\notin T$ and
	{$ \mycal{CT} \models \exists(B) \wedge \forall (B \implies \exists x'($
		$chr\left(H_{k}\right)=(H'_{k})\wedge chr\left(H_{r}\right)=(H'_{r}) \wedge g))$}.
In this case, $S_{chr_{vis}}$ has the form: $\langle [c\#i:j|G], H_{k}\uplus H_{r}\uplus S,Gr,B,T,H\_ann\rangle_{m}$. The output state $S_{chr_{vis}}'$ has the form 
\\{$\langle C ++ H ++ G , H_{k}\cup S ,Gr,B \wedge chr\left(H_{k}\right)=(H'_{k})\wedge chr\left(H_{r}\right)=(H'_{r}) \wedge g,$}
		{$T \cup \{ \langle r,id\left(H_{k}\right) + id\left(H_{r}\right)\rangle\}, H\_ann \rangle_{m} $.}
		Due to the fact that $S_{chr}$ is equivalent to $S_{chr_{vis}}$,it has the following form: $\langle G, H_{k}\uplus H_{r}\uplus S,B,T\rangle_{n}$. For the same program, the \textsf{CHR} rule $r$ is applicable
		{producing $S_{chr}'$:}\\{$\langle C ++ G , H_{k}\cup S ,chr\left(H_{k}\right)=(H'_{k})\wedge chr\left(H_{r}\right)=(H'_{r}) \wedge g \wedge B,$}
		$T \cup \{ \langle r,id\left(H_{k}\right) + id\left(H_{r}\right) \rangle_{n} \}$
		 \[ H =
  \begin{cases}
    \left[c\#i:j\right]      & \quad \text{if } c \text{ occurs in }H'_k\\
   \left [\;\right] & \quad \text{otherwise}\\
  \end{cases}
\]
\\
{Due to the fact that the same \textsf{CHR} rule is applied for both states, the new built-in stores are equivalent according to Definition} \ref{builtinequiv}. {This is due to the fact that since the original states have equivalent constraint stores,} {we assume without loss of generality} {that the matchings in both cases are the same since the same rule was applied. Thus, the rule in the two programs $P_{chr}$ and $P_{chr_{vis}}$ are renamed similarly.}
\\At this point, for $S_{chr_{vis}}'$ one of two cases is possible:
\begin{itemize}
	\item \textbf{An annotation rule is applicable:}
	\\In this case, $S_{chr_{vis}}'$ has the form $\langle \left[cons\#id:occ|A\right] , Head \uplus St,Gra,Bu,T_H,H\_ann\rangle_{m}$ where there exists a renamed, constraint annotation rule {with the same variables $x'$} of the form:
		$g\; ru\; @\; H'\;$ 
		$==> Condition\; |\; Obj'$ where $cons$ is part of $H'$ 
		\\such that
		 $\mycal(CT) \models \exists\left(B \right) \wedge \forall (B \implies \exists x' \wedge Condition \wedge chr \left(Head\right)=(H') \wedge output\_graphical\_object\left(H', x', Obj'\right)=Obj)$ 
		and $ \neg contains\left(H\_ann ,\left( ru, 
		id\left(Head\right) \right)\right)$.
		The output state $S_{chr_{vis}}''$ has the form: 
		\\$\langle \left[Obj\#\langle ru,id\left(Head\right),\{\;\} \rangle,cons\#id:occ|A\right],Head \uplus St,Gra,$ 
		$Bu,$
		$T\_H , H\_ann \cup \{\langle r, id\left(H\right) \{\;\} \rangle\} \rangle_{m}$.
		\\Similar to the previous cases, $S_{chr_{vis}}''$ triggers the transition \emph{draw} or \emph{update store} producing $S_{chr_{vis}}''':\langle \left[cons\#id:occ|A\right],Head \uplus St,Gra',$ 
		$Bu,$
		$T\_H , H'\_ann \rangle_{m'}$.
		According to Definition \ref{def:eqchrvis}, $S_{chr_{vis}}'''$ is equivalent to $S_{chr}'$ since the two stacks, constraint stores and propagation histories are not affected by the application of the annotation rule.
		\item \textbf{No annotation rule is applicable}
		\\At this point, $S_{chr_{vis}}'$ is still equivalent to $S_{chr}'$
		\\
\end{itemize}
\item \textbf{Case 6: Applying Drop}
\\In the case where $S_{chr_{vis}}=$ $\langle\left[c\#i:j|A\right],  S,Gr,B,T,H\_ann\rangle_{m}$ 
such that $c$ has no occurrence $j$ in the program and case 5 is not applicable, the transition \emph{Drop} is triggered. \emph{Drop} produces the state $S_{chr_{vis}}'=$
			$
		\langle A,  S,Gr,B,T,H\_ann\rangle_{m} 
			$
		\\given that $c\#i:j$ is an occurrenced active constraint and $c$ has no occurrence $j$ in the program. 
		\\
		Since $S_{chr}$ is equivalent to $S_{chr_{vis}}$, they both have the same stack $\left[c\#i:j|A\right]$. Thus under $\omega_{vis}$, the same transition \emph{drop} is triggered producing $S_{chr}': \langle A,  S,B,T\rangle_{n}$. According to Definition \ref{def:eqchrvis}, $S_{chr_{vis}}'$ and $S_{chr}'$ are equivalent as well.
		\\
		\item \textbf{Case 7: Applying Default}
		\\In the case where none of the above cases hold, the transition \emph{Default} transforms $S_{chr_{vis}}$ to
		\\$
		S_{chr_{vis}}': \langle \left[c\#i:j+1|A\right],  S,Gr,B,T,H\_ann\rangle_{m} 
			$.
			Similarly the equivalent state $S_{chr}$ triggers the same transition \emph{Default} in this case. The output state $S_{chr}': \left[c\#i:j+1|A\right],  S,Gr,B,T,H\_ann\rangle_{n} $ is still equivalent to $S_{chr_{vis}}'$
\end{enumerate}
\end{proof}
\end{theorem}

\subsubsection{Operational Semantics including Rule Annotations}
\(\)\\Table \ref{table:omegavisruleann} shows the operational semantics of $\omega_{vis_{r}}$. $\omega_{vis_{r}}$ is basically $\omega_{vis}$ but taking rule annotations into account. A state of $\omega_{vis_{r}}$  is a tuple $\langle G,S,Gr,B,T, H\_ann,Cons\_r \rangle_{n}$. $G$, $S$, $Gr$, $B$, $T$, $H\_ann$, and $n$ have the same meanings as in an $\omega_{vis}$ state. However, $G$ {can hold ,in addition to the previously seen formats of constraints, the constraint in addition} {the name of the rule that added it}. $Cons\_r$ holds for each constraint, its ID in addition to the name of the rule used to add it. It is an extended form of the constraint store.
{If the constraint comes from the query of the user or if }{it is an auxiliary constraint,} {the keyword, \verb+aux+ is used for the rule name inside $G$ and $Cons\_r$}.
\begin{definition}
\label{def:rulename}
\(\\\)Similar to the previously defined functions, 
for a $\omega_{vis_{r}}$state $\langle G, S,Gr,B,T,H\_ann,Cons\_r\rangle_{n}$, $rule\_name\left(H\right)$ is a function defined as: $rule\_name\left(c\#n\right)=r$ such that $c\#n\#r$ $\in$ $Cons\_r$
\end{definition}

\begin{definition}
\label{def:newget}
\(\\\)The definition of the function $get\_constraints\left(Sq\right)$ is modified such that:
for a sequence $Sq=\left(c_1\#r_1,\ldots,c_n\#r_n\right)$ the function $get\_constraints\left(Sq\right)=\cup_{i=1}^{n} c_{i}$ such that $c_{i}$ is not an auxiliary constraint.
{$r_1\ldots r_n$ could thus represent identifiers or rule names.}
{Definition \ref{def:newget} is thus reflected in the equivalence definition of states introduced in Definition \ref{def:eqchrvis}
}
\end{definition}
\begin{longtable}{l}
    \hline
    \\
    1. \textbf{Solve+wakeup}:
			\\$\langle \left[c|A\right],  S_{0}\uplus S_{1},Gr,B,T,H\_ann,H\_cons\rangle_{n} 
			\mapsto_{solve+wake}
			 \langle S_{1}++A, S_{0}\uplus S_{1},Gr,B',T,H\_ann,H\_cons\rangle_{n}$
			\\given that $c$ is a built-in constraint and $\mycal{CT} \models \forall((c\wedge B \leftrightarrow B'))$
			and $wakeup\left(S_{0}\uplus S_{1} ,c,B\right) = S_{1}$ \\
    \\\hline\\
    	2. \textbf{Activate}
		$\langle\left[c\#rule\_name|A\right],  S,Gr,B,T,H\_ann,H\_cons\rangle_{n} 
			\mapsto_{activate}$
			\\$\langle \left[c\#n:1|A\right],  \{c\#n\} \cup S,Gr,B,T,H\_ann,H\_cons\cup \{c\#n\#rule\_name\}\rangle_{n+1}
			$
		\\given that $c$ is a \textsf{CHR} constraint.\\
        \\\hline\\
 3. \textbf{Reactivate}
		$\langle \left[c\#i|A\right],  S,Gr,B,T,H\_ann,H\_cons\rangle_{n} 
			\mapsto_{activate}
			\langle \left[c\#i:1|A\right],   S,Gr,B,T,H\_ann,H\_cons\rangle_{n}
			$
		\\given that $c$ is a \textsf{CHR} constraint and $i$ is not a valid rule name.
		\\
        \\\hline\\
			4. \textbf{Draw}
		: $\langle \left[Obj\#\langle r, id\left(H\right),Obj\_ids \rangle|A\right], S,Gr,B,T,H\_ann,H\_cons\rangle_{n} \mapsto_{draw} \langle A,S ,Gr\cup \{Obj\#n\},B,T$\\$,H'\_ann,H\_cons\rangle_{n+1}$ given that $Obj$ is a graphical object: $graphical\_object\left(Actual_0,\ldots,Actual_k\right)$ and
		\\$H'\_ann=generate\_new\_ann\_history\left(Obj,n,id\left(H\right),H\_ann\right)$
		\\The actual parameters of $graphical\_object$ are used to visually render the object.\\
		  \\\hline\\
		  5. \textbf{Update Store}
		: $\langle \left[Obj\#\langle r,id\left(H\right),Obj\_ids\rangle|A\right], S,Gr,B,T,H\_ann,H\_cons\rangle_{n} \mapsto_{update\;store} \langle A,S ,Gr'$
		\\$,B,T,H\_ann,H\_cons\rangle_{n}$ given that $Obj$ is a graphical action: $graphical\_action\left(Actual_0,\ldots,Actual_k\right)$.
		\\$Gr'=update\_graphical\_store\left(Gr,graphical\_action\left(Actual_0,\ldots,Actual_k\right)\right)$
		\\The actual parameters of $graphical\_action$ are used to update the graphical objects.\\
		  \\\hline\\
		6. \textbf{Apply\_Annotation}:
		\\$\langle \left[c\#i:j|G \right],H \cup S,Gr,B,T,H\_ann, H\_cons\rangle_{n} \mapsto_{apply\_annotation}$ \\$\langle \left[Obj\#\langle r,id\left(H\right),\{\;\}\rangle,c\#i:j|G\right], H \cup S ,Gr ,
		B, T , H\_ann \cup \{(ann\_name, id(H))\}, H\_cons\rangle_{n}$ \\
		where there is a, renamed, constraint annotation rule {with variables $y'$} of the form: 
		\\$g\; ann\_name\; @\; H'$ 
		$==> Condition | Obj'$ 
		where $c$ is part of $H'$
		such that
			\\{$\mycal(CT) \models \exists\left(B\right) \wedge \forall (B \implies \exists y' $ $(chr\left(H\right)=(H') \wedge output\_graphical\_object\left(H',y', Obj'\right)=Obj))$ 
		}
		\\ and $rule\_name\left(H\right)$ does not have an associated annotation rule
		and $\neg contains\left(H\_ann,\left(r,id\left(H\right)\right)\right)$
		\\\hline\\
		7. \textbf{Apply}
		: $\langle \left[c\#i:j|A\right], H_{k}\uplus H_{r}\uplus S,Gr,B,T,H\_ann,$ 
		$H\_cons\rangle_{n} \mapsto_{apply}$
		\\
		$\langle C' ++ H ++  A , H_{k}\cup S ,Gr',chr\left(H_{k}\right)=(H'_{k})\wedge chr\left(H_{r}\right)=(H'_{r}) \wedge g$ 
		$\wedge B,T \cup \{ \langle r,id\left(H_{k}\right) + id\left(H_{r}\right) \rangle \},$
		\\$H\_ann, H\_cons\rangle_{n}$ \\
		 there is and a renamed rule in $P_{vis}$ with variables $x'$ where the jth occurrence of $c$ is part of the head. \\The renamed rule has the form:\\
		$r\; @\; H'_{k}\;\backslash \; H'_{r} \Leftrightarrow \; g \; | \; C.$\\
		such that 
		{$ \mycal{CT} \models \exists(B) \wedge \forall (B \implies \exists x'$ $(chr\left(H_{k}\right)=(H'_{k})$ $\wedge$  $chr\left(H_{r}\right)=(H'_{r})$  $\wedge g$))}
		\\and $\langle r,id\left(H_{k}\right) + id\left(H_{r}\right) \rangle\notin T$
		and there are no associated and applicable annotation rule(s) to $c$ 
		\\or any part(s) of it that are not executed yet.
		\\$C'=C\#r$ 
		 if and only if there is no annotation rule associated with $r$
		\\{otherwise if $r$ has an associated rule under the same variables (matching) $x'$:}
		\\$g\;rule\_name\;r==> ann\_cond | Aux\_cons'$
		\\such that :
		\\{$\mycal(CT) \models \exists(B) \wedge \forall (B \implies (Aux\_cons') = Aux\_cons \wedge ann\_cond )$}
		\\or\\
		{$\mycal(CT) \models \exists(B) \wedge \forall (B \implies \exists x' ((chr(H_{k})=(H'_{k}) \wedge chr(H_{r})=(H'_{r}) \wedge g \wedge (Aux\_cons') = Aux\_cons \wedge ann\_cond ))$}
		\\then $C'=\left[Aux\_cons\#aux | C\#r\right] $
		\\If $c$ occurs in $H'_k$ then $H=\left[c\#i:j\right]$ otherwise $H=\left[\right]$.
		\\If the program communicates the head constraints (i.e. contains \verb+comm_head(T) ==> T=true+) then
		\\${Gr'= remove\_gr\_obj\left(G,id\left(H_{r}\right),H\_ann\right)}$  otherwise $Gr'=Gr$.
		\\
		 \\\hline\\
				8. \textbf{Drop}
		\\
		$\langle\left[c\#i:j|A\right],  S,B,T,H\_ann,H\_cons\rangle_{n} 
			\mapsto_{drop}
		\langle A,  S,B,T,H\_ann,H\_cons\rangle_{n}$
		\\given that $c\#i:j$ is an occurrenced active constraint and $c$ has no occurrence $j$ in the program
		\\and that there is no applicable constraint annotation rule for the constraint c.
		\\
		 \\\hline
		\\
		 9. \textbf{Default}
		$\langle\left[c\#i:j|A\right],  S,B,T,H\_ann,H\_cons\rangle_{n} 
			\mapsto_{drop}
		\langle \left[c\#i:j+1|A\right],  S,B,T,H\_ann,H\_cons\rangle_{n}$
		\\
		in case there is no other applicable transition.
		\\\\
		 \hline		
\caption{Transitions of $\omega_{vis}$ taking rule annotations into account}
	\label{table:omegavisruleann}
\end{longtable}
As seen through Table \ref{table:omegavisruleann}, the transition \emph{Apply} could activate a rule annotation (if applicable) associated with the applied \textsf{CHR} rule. In this case, the auxiliary constraint is added to the goal with the keyword $aux$.
Since the \textsf{CHR} rule is applied with the combination of constraints once, the rule annotation is also applied once and there is no risk of running infinitely. The rest of the transitions were modified only to include the newly added state parameter $H\_cons$.
In addition, the \emph{apply annotation} transition is only applied if the constraints activating it were not added by a rule that is associated with a rule annotation rule. This ensures the previously discussed property that whenever a rule is associated with a rule annotation, then the body constraints can never trigger their own visual annotation rules. Since the rule is annotated then the individual constraint annotations are discarded.

\subsubsection{Completeness and Soundness}
\(\)\\Since according to the modified definition shown in Definition \ref{def:newget}, the state equivalence will not take auxiliary constraints into account.
The proofs for completeness and soundness shown in Proof \ref{pr:pro1} and Proof \ref{pr:pro2} still hold.
As seen from the transitions in Table \ref{table:omegavisruleann}, the auxiliary constraint of the rule annotation is dealt with as a normal \textsf{CHR} constraint. Thus the only difference between these transitions and the previous ones is that sometimes, auxiliary constraints will exist in the constraint store of the final states. The states however remain equivalent.
Thus even if annotations of rules were applicable, $\omega_{vis_{r}}$ is still sound and complete.

Proof \ref{proof:pr1} introduced for the completeness check has the following amendment in the induction step:
\\\textbf{Case 4: (Applying the transition Apply)}
		\\For $\omega_{r}$, the transition \emph{Apply} is triggered in the case where $S_{chr}=\langle \left[c\#i:j|A\right], H_1 \uplus H_2 \uplus S,B,T\rangle_{n} $
			such that the jth occurrence of $c$ is part of the head of the re-named apart rule {with variables $x'$}:
		{\(r\; @\; H'_1\; \backslash \;H'_2 \; \Leftrightarrow \;g \;|\; C.\)}
		\\
		$ (chr\left(H_1\right)=(H'_{1}) \wedge chr\left(H_{2}\right)=(H'_{2}) \wedge$
	{	$ \mycal{CT} \models \exists(B) \wedge \forall (B \implies \exists x' (chr\left(H_1\right)=(H'_{1}) \wedge chr\left(H_{2}\right)=(H'_{2}) \wedge g))$ }
		and $\left(r,id\left(H_1\right)+id\left(H_2\right)\right) \notin T$.
		In this case, $S_{chr} \mapsto_{apply\;r}$
			$S_{chr}':\langle C + H + A,  H_1 \cup S,$
			{$chr\left(H_1\right)=(H'_{1}) \wedge chr\left(H_{2}\right)=(H'_{2}) \wedge g \wedge B,$}
			$T \cup \{\left(r,id\left(H_1\right)+id\left(H_2\right)\right)\}\rangle_{n}$
		 \[ H =
  \begin{cases}
    \left[c\#i:j\right]      & \quad \text{if } c \text{ occurs in }H'_1\\
    \left[\;\right] & \quad \text{otherwise}\\
  \end{cases}
\]
\\Due to the fact that $S_{chr}$ and $S_{chr_{vis}}$ are equivalent, they both have the same goal stacks and constraint stores. However for $\omega_{vis}$ one of two possibilities could take place.
		\\\tabitem It could be that there is no applicable constraint constraint annotation rule either because  any applicable annotation rule was already executed or at this point there is no applicable annotation rule at this point.
		The transition $Apply$ is thus triggered right away under $\omega_{vis}$.
		\\
		\\\tabitem It could, however, be the case that there is an applicable constraint annotation rule.
		\\In this case an annotation rule ($r_{ann}$) for $c$ is applicable such that:
		\\$S_{chr_{vis}}\langle \left[c\#i:j|A\right],H_1 \uplus H_2 \uplus S,Gr,B,T,H\_ann,H\_cons\rangle_{m} \mapsto_{apply\_annotation}$ 
		\\$S_{chr_{vis}}'':\langle \left[Obj\#\langle  r,id\left(H\right),\{\} \rangle,c\#i:j|A\right],H_1 \uplus H_2 \uplus S ,Gr , B,$ 
		$T , H\_ann \cup \{\langle r_{ann}, id\left(H\right),\{\;\}\rangle\},h\_cons \rangle_{m}$ .
		\\At this point either the transition \emph{Draw} or \emph{Update store} should be applied.
		Thus, 
		\\$S_{chr_{vis}}' \mapsto_{draw\big/update store} S_{chr_{vis}}':\langle \left[c\#i:j|A\right], H_1 \uplus H_2 \uplus S ,Gr', B,$ 
		$T , H'\_ann, H\_cons \rangle_{m'}$
		\\The transition \emph{Draw} is applicable in case $Obj$ is a graphical object such that:
		$Gr'=Gr\cup\{Obj\#m\}\;\wedge\;m'=m+1 \;\wedge\; H'\_ann = generate\_new\_ann\_history\left(Obj,m,r,id\left(H\right),\cup \{\langle r_{ann}, id\left(H\right),\{\;\}\rangle\}\right)$.
		\\In the case where $Obj$ is a graphical action, the transition \emph{Update Store} is applied such that:
		\\$Gr'=update\_graphical\_store\left(Gr,Obj\right)\;m'=m\;H'\_ann=\cup \{\langle r_{ann}, id\left(H\right),\{\;\}\rangle\}$
		\\\\\\The two transitions do not affect the constraint stores or goal stacks. Thus, the equivalence of the states is not affected.
		\\Due to the fact that $S_{chr}$ and $S_{chr_{vis}}'$ are equivalent, in the case where $S_{chr}$ triggers the transition \emph{Apply} under $\omega_{r}$, $S_{chr_{vis_{r}}}$ also triggers the same transition for the same \textsf{CHR} rule, under $\omega_{vis}$ producing a state 
		($S_{chr_{vis}}'':\langle C' + H + A,  H_1 \cup S,Gr'',$
			{$chr(H_1) = H'_1 \wedge chr(H_2) = H'_2 \wedge g \wedge B,$}
			$T \cup \{(r,id(H_1)+id(H_2))\}, H\_ann, H\_cons\rangle_{m}$)
			where 
			\\{$ \mycal{CT} \models \exists(B) \wedge \forall (B \implies \exists x' (chr(H_1)=(H'_{1}) \wedge chr(H_{2})=(H'_{2}) \wedge g))$
			}
			and
					 \[ H =
  \begin{cases}
    \left[c\#i:j\right]      & \quad \text{if } c \text{ occurs in }H'_1\\
    \left[\;\right] & \quad \text{otherwise}\\
  \end{cases}
\]
			\\{Since the same rule is applied in both cases, the two resulting states
			$S_{chr_{vis}}''$ and ($S_{chr}'$) are equivalent.}
			{Similarly, the matchings in both cases are have to be the same since the original states have equivalent stores. Thus. without loss of generality, the applied rule is assumed to be renamed with the same variables in both programs.} {The two new built-in stores are still equivalent and the two resulting states are thus also equivalent.}
			\\If the program communicates the head constraints (i.e. contains \verb+comm_head(T) ==> T=true+) then
		\\${Gr''= remove\_gr\_obj\left(Gr',id\left(H'_{r}\right),H'\_ann\right)}$
			There are however, two possibilities in this case, 
			\begin{itemize}
			    \item if $r$ is associated with an annotation rule {(renamed with variables $x'$)}
			    \\$r ==> ann\_cond | Aux\_cons'$.
			    \\where 
		 $ \mycal{CT} \models \exists(B) \wedge \forall (B \implies (Aux\_cons')=Aux\_cons \wedge ann_cond$ 
			    \\In this case, $C'=[Aux\_cons\#aux|C]$.
			    \\$S_{chr_{vis}}'$ is thus still equivalent to $S_{chr}'$ since auxiliary constraints are disregarded in the equality check.
			    
			     \item if $r$ is not associated with an annotation
			     or if the annotation is not applicable then $C'=C\#r$ making the two states $S_{chr_{vis}}'$ and $S_{chr}'$ equivalent as well.
			\end{itemize}
			
As for the soundness proof (Proof \ref{pro:proof2}),
the only change is in the following case:
\\\textbf{Case 5: the Apply transition}
\\
In the case where a \textsf{CHR} rule is applicable to $S_{chr_{vis}}$, the transition 
\emph{Apply} is triggered under $\omega_{vis}$. A \textsf{CHR} rule $r$ is applicable in the case where a renamed version of the rule $r$ {with variables $x'$} is
\\($r\; @\; H'_{k}\;\backslash \; H'_{r} \Leftrightarrow \; g \; | \; C.$) where $\langle r,id\left(H_{k}\right) + id\left(H_{r}\right) \rangle\notin T$ 
and 
\\{$ \mycal{CT} \models \exists(B) \wedge \forall (B \implies \exists x' (chr\left(H_r\right)=(H'_{r}) \wedge chr\left(H_{k}\right)=(H'_{k}) \wedge g))$
			}.
In this case, $S_{chr_{vis}}$ has the form: $\langle \left[c\#i:j|G\right], H_{k}\uplus H_{r}\uplus S,Gr,B,T,H\_ann,H\_cons\rangle_{m}$. The output state $S_{chr_{vis}}'$ has the form {$\langle C' + H + G , H_{k}\cup S ,Gr',$
\\$chr\left(H_r\right)=(H'_{r}) \wedge chr\left(H_{k}\right)=(H'_{k}) \wedge g \wedge B ,$}
		$T \cup \{ \langle r,id\left(H_{k}\right) + id\left(H_{r}\right) \},H\_ann, H\_cons\rangle_{m}$.
		Due to the fact that $S_{chr}$ is equivalent to $S_{chr_{vis}}$,it has the following form: $\langle G, H_{k}\uplus H_{r}\uplus S,B,T\rangle_{n}$. For the same program, the \textsf{CHR} rule $r$ is applicable producing $S_{chr}'$: \\{$\langle C ++ G , H_{k}\cup S ,
		chr\left(H_r\right)=(H'_{r}) \wedge chr\left(H_{k}\right)=(H'_{k}) \wedge g \wedge B,$}
		$T \cup \{ \langle r,id\left(H_{k}\right) + id\left(H_{r}\right) \rangle \}$
		such that
\\		{$ \mycal{CT} \models \exists(B) \wedge \forall (B \implies \exists x' (chr\left(H_r\right)=(H'_{r}) \wedge chr\left(H_{k}\right)=(H'_{k}) \wedge g))$
			}.
		{Similarly, without loss of generarility the same variable renaming was used for both programs. Thus, the two new built-in stores are equivalent according to Definition} \ref{builtinequiv} {since only one possible matching would be possible since the original two states have equivalent stores.}
		Moreover, 
		 \[ H =
  \begin{cases}
    \left[c\#i:j\right]      & \quad \text{if } c \text{ occurs in }H'_k\\
    \left[\;\right] & \quad \text{otherwise}\\
  \end{cases}
\].
In addition, 
		 \[ Gr' =
  \begin{cases}
    remove\_gr\_obj\left(Gr,id\left(H'_{r}\right),H\_ann\right)     & \quad \text{if the program communicates the head constraints to the visual tracer}\\
    Gr & \quad \text{otherwise}\\
  \end{cases}
\].
If a rule annotation ($r ==> ann\_cond | Aux\_cons'$) is applicable, then $C'=\left[Aux\_cons|C\right]$ where $(Aux\_cons')=Aux\_cons$ otherwise $C'=C\#r$.
\\ In both cases, the output states are still equivalent since the equivalence check neglects auxiliary constraints and the graphical store.

\section{$CHR^{vis}$ to $CHR^{r}$ Transformation Approach}
\label{sec:trans}
The aim of the transformation is to eliminate the need of doing any compiler modifications in order to animate \textsf{CHR} programs. 
A $CHR^{vis}$ program $P^{vis}$ is thus transformed to a corresponding $CHR^{r}$ program $P$ with the same behavior. $P$ is thus able to produce the same states in terms of \textsf{CHR} constraints and visual objects as well.

As a first step, the transformation adds for every constraint \verb+constraint/n+ a rule of the form:
\\\(
comm\_cons\_{constraint}\;@\;constraint\left(X_{1},X_{2},...,X_{n}\right)\; \Rightarrow\; check\left(status, false\right)\; |\; \)
                \\\hspace*{3.7cm}\(communicate\_constraint\left(constraint\left(X_{1},X_{2},...,X_{n}\right)\right).\)

The extra rule ensures that every time a \verb+constraint+ is added to the store, the tracer (\emph{external module}) is notified. If \verb+constraint+ was annotated as an interesting constraint, its corresponding annotation rule is activated producing the corresponding visual object(s). The new rules communicate any \verb+constraint+ added to the constraint store.

The user can also choose to communicate to the tracer the head constraints since they could affect the animation. A removed head constraint could affect the visualization in case it is an interesting constraint. In this case, if the user chose to communicate head constraints, the associated visual object, produced before, should be removed from the visual trace.\footnote{The tracer is able to handle the problem of having multiple Jawaa objects with the same name by removing the old object having the same name before adding the new one. This is possible even if the removed head constraint was not communicated.}.

As a second step, the transformer adds for every compound constraint-annotation of the form:\\$cons_{1},\ldots,cons_{n}==>annotation\_constraint_{cons_{1},\ldots,cons_{n}}\left(Arg{1},\ldots,Arg_{m}\right)$, a new rule of the form:
\\$ compound_{cons_{1},\ldots,cons_{n}}\; @\; cons_{1}\left(Arg_{cons_{1_{1}}},\ldots,Arg_{cons_{1_{1x}}}\right), \ldots,cons_{n}\left(Arg_{cons_{n_{1}}},\ldots,Arg_{cons_{n_{ny}}}\right)$
\\$\Rightarrow check\left(status, false\right)\;|\; annotation\_constraint_{cons_{1},\ldots,cons_{n}}\left(Arg{1},\ldots,Arg_{m}\right)$.

By default, a propagation rule is produced to keep $cons_{1},\ldots,cons_{n}$ in the constraint store. However, the transformer could be instructed to produce a simplification rule instead. The annotation is triggered whenever $cons_{1},\ldots,cons_{n}$ exist in the constraint store. Whenever this is the case, the rule $compound_{cons_{1},\ldots,cons_{n}}$ is triggered producing the annotation constraint. Since the annotation constraint is a normal \textsf{CHR} constraint, it is automatically communicated to the tracer using the previous step.

As a third step, the \textsf{CHR} rules annotated by the user as interesting rules should be transformed. The idea is that the \textsf{CHR} constraints produced by such rules should be ignored. In other words, even if the rule produces an interesting \textsf{CHR} constraint, it should not trigger the corresponding constraint annotation. Instead, the rule annotation is triggered.

Hence, to avoid having problems with this case, a generic \emph{status} is used throughout the transformed program $P_{Trans}$.
Any rule annotated by the user as an interesting rule changes the \verb+status+ to $true$ at execution. However, the rules added in the previous two steps check that the status is set to $false$. In other words, if the interesting rule is triggered, no constraint is communicated to the tracer since the guard of the corresponding $communicate\_constraint$ rule fails.
Any rule $rule_{i} @ H_{K}\;\backslash\;H_{R}\; \Leftrightarrow\; G \;| \; B $ with the corresponding annotation $rule_{i}==>annotation\_constraint_{rule_{i}}$ is transformed to:
$rule_{i} @ H_{K}\;\backslash\;H_{R}\; \Leftrightarrow\; G \;|\;set\left(status,true \right),\; B,\;annotation\_constraint_{rule_{i}},\\\;set\left(status,false \right). $
In addition, the transformer adds the following rule to $P_{Trans}$:
\\$comm\_cons_{annotation\_constraint_{rule_{i}}}\;@\; annotation\_constraint_{rule_{i}} \;\Leftrightarrow$\\\hspace*{3.7cm}$\;communicate\_constraint\left(annotation\_constraint_{rule_{i}} \right)$.
\\The new rule thus ensures that the events associated with the rule annotation are considered and that all annotations associated with the constraints in the body of the rule are ignored.

\subsection{Correctness of Transformation Approach}
The aim of the transformation process is to produce a $CHR^{r}$ program ($P_{trans}$) that is able to perform the same behavior of the corresponding $CHR^{vis}$ program ($P_{vis}$) which basically contains the original \textsf{CHR} program $P$ along with the constraint(s) and rule annotations. This section shows that the transformed program, using the steps shown previously, is a correct one. In other words, for the same query $Q$, $P_{trans}$ produces an equivalent state to the one produced by $P$. As seen from the previous section $\omega_{vis}$ was proven to be sound and complete. This implies that any state reachable by $\omega_{r}$ is also reachable by $\omega_{vis}$. In addition, any state reachable by $\omega_{vis}$ is also reachable by $\omega_{r}$. The focus of this section is the initial \textsf{CHR} program provided by the user. The aim is to make sure that $P_{trans}$  produces the same \textsf{CHR} constraints that $P$ produces to make sure that the transformation did not change the behavior that was initially intended by the programmer. The focus is thus to compare how $P$ and $P_{trans}$ perform over $\omega_{r}$.

\begin{theorem}
Given a \textsf{CHR} program $P$ (along with its annotations) and its corresponding transformed program $P_{trans}$ and two states $S_{1} = \langle G, \phi \rangle$ and $S_{2} = \langle G, \phi \rangle$ where $G$ contains the initial goal constraints and $Aux'$ is a set of auxiliary constraints. Then the following holds: 
\\If
\( S_{1} \xmapsto[\omega_{r}]{P} \mathrel{\vphantom{\to}^*} S_{1}^{'}\)
and \(S_{2}\xmapsto[\omega_{r}]{P_{trans}} \mathrel{\vphantom{\to}^*} S_{1}^{'} \cup Aux'\)
 then $P_{trans}$ is equivalent to $P$

\end{theorem}
\begin{proof}(Sketch)
The below sketch shows how the constraint store changes over the course of running the same query $Q$ in $P$ and $P_{trans}$.
\\Executing a query in $P$ takes the following steps. The store at each step starts with the final value of the previous one
\begin{enumerate}
	\item Step 1: The constraint store of $P$ ($S_{P}$) start off by being empty ($\phi$).
	$S_{P_{step 1}}=\phi$.
	\item Step2: The constraints in the query start to get activated and enter the constraint store $S_{P}$. $S_{P_{step 2}}$ starts off by being equal to $S_{P_{step 1}}$. Every time a constraint $c$ gets activated, it enters the store converting $S_{P_{step 2}}$ to $S_{P_{step 2}}\cup c$.
		\item  Step 3:
Applying a rule $ r @ H_{k} \backslash H_{r}\; \Leftrightarrow \; Gu  \; | \;B$, adds $B$ to $S_{P}$. The components of $H_{r}$ are also removed from it. 
Thus after the rule application, $S_{P_{step 3}} = S_{P_{step 2}} \cup B - H_{r}$. 
\item Step 4: Go Back to step 3. Any applicable rule is fired changing the constraint store. Step 3 keeps on repeating until no more rules are applicable reaching a fixed point. 
\end{enumerate}
In $P_{trans}$ executing the same query undergoes the following steps. Similarly, the store at each step starts with the final value of the previous one
\begin{enumerate}
\item Step 1: Initially the constraints store is also empty. Thus, $S_{P_{trans_{step 1}}}=\phi$.
\item Step 2: The query constraints start to get activated and enter the constraint store $S_{P_{trans}}$. $S_{P_{trans_{step 2}}}$ starts off by being equal to $S_{P_{trans_{step 1}}}$. When a constraint $c$ gets activated, it enters the store. $S_{P_{trans_{step 2}}}$ thus is augmented with $c$ and the result is $S_{P_{trans_{step 2}}} = S_{P_{trans_{step 2}}}\cup c$.
\item Step 3: In $P_{trans}$, each time a constraint gets activated and enters the store, the corresponding rule: $comm\_cons\_{constraint}$ is triggered. $comm\_cons\_{constraint}$ is a propagation rule. Thus it does not remove any constraint from the store. The body of the rule communicates the constraint to the tracer, not adding anything to the store as well.
At this point $S_{P_{trans_{step 3}}}=S_{P_{trans_{step 2}}}$.
\item Step 4 (Optional): In $P_{trans}$, if there are any compound annotations for constraints ($cons_{1},\ldots,cons_{n}$) in $S_{P_{trans}}$, the corresponding rule $compound_{cons_{1},\ldots,cons_{n}}$ is fired adding to the $S_{P_{trans}}$ the auxiliary constraint: \\$annotation\_constraint_{cons_{1},\ldots,cons_{n}}\left(Arg{1},\ldots,Arg_{m}\right)$.
\\$annotation\_constraint_{cons_{1},\ldots,cons_{n}}\left(Arg{1},\ldots,Arg_{m}\right)$ fires the corresponding $comm\_cons\_{constraint}$ rule.
Thus, $S_{P_{trans_{step 4}}}=S_{P_{trans_{step 3}}}\cup Aux_{cons}$ where $Aux_{cons}=annotation\_constraint_{cons_{1},\ldots,cons_{n}}\left(Arg{1},\ldots,Arg_{m}\right)$.
\item Step 5: $P$ contains the same \textsf{CHR} rules as in $P_{trans}$. Consequently, whenever a rule is applicable in $P$, the same rule will be also applicable in $P_{trans}$.
Applying a rule $ r @ H_{k} \backslash H_{r}\; \Leftrightarrow \; Gu  \; | \;B$, adds $B$ to $S_{P}$. The components of $H_{r}$ are also removed from the store. 
Thus after the rule application, $S_{P_{trans_{step 5}}} = S_{P_{trans_{step 4}}} \cup B - H_{r}$. 
In case $r$ has an associated annotation rule,  Once the constraints in $B$ are added to the store, the corresponding rules $comm\_cons\_{constraint}$ are matched. However, since their guard fails (as the status is set to true in $P_{trans}$, the rules do not file.
In case $r$ has an associated annotation rule, the auxiliary constraint $annotation\_constraint_{r}$ is added to the store. Afterwards, $annotation\_constraint_{r}$ triggers the simplification rule $comm\_cons_{annotation\_constraint_{rule_{i}}}$ communicating the $annotation\_constraint_{r}$ constraint to the tracer. Since the rule is a simplification rule, $comm\_cons_{annotation\_constraint_{rule_{i}}}$ is then removed from the constraint store keeping $S_{P_{trans_{step 5}}} = S_{P_{trans_{step 4}}} \cup B - H_{r} $. 
\item Step 3 and 4 could be applied for the new constraint store. As sen previously, step 3 does not change the constraint store. Step 4 could however add some auxiliary constraints to $S_{P_{trans_{step 5}}}$.
\end{enumerate}
\end{proof}
As seen before, in all the previous steps, either $S_{P_{trans}}= S_{P}$ or $S_{P_{trans}}= S_{P}\cup Aux$ where $Aux$ contains some extra auxiliary constraints. Thus, the transformation does not change in the intended application of $P$.

\section{Applications}
\label{sec:app}
This section introduces different applications of the proposed semantics.
The applications cover different fields were animations were useful.
\subsection{Animating Java Programs}
The idea of annotating constraints with visual objects was extended to Java programs in \cite{DBLP:conf/iv/SharafAF16}. The execution path of the new programs is shown in Figure \ref{fig:javaexecpath}.
\begin{figure}%
\centering
\includegraphics[width=70mm]{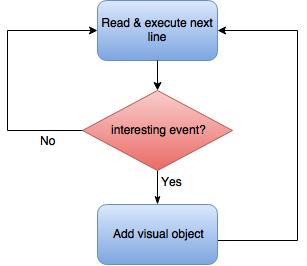}%
\caption{New Execution Path}%
\label{fig:javaexecpath}%
\end{figure}
Using visual annotation rules of \textsf{CHR} programs, Java programs could also be animated in a generic way.
The objective is to annotate method calls as interesting events. They are thus linked with visual objects. Thus, every time a method is invoked, its corresponding visual object is added. This results in animating the algorithm while it is running.
That way, the user does not have to get into any of the technical details of how the visualization is produced. They only need to state what they want to see.

For every interesting method $m\left(arg_1,\ldots,arg_n\right)$, one \textsf{CHR} rule is produced.
The rule simply communicates the fact the the method $m$ with the arguments $arg_1,\ldots,arg_n$ was called in the Java program.
The rule has the following format:
\begin{Verbatim}
            m(arg_1,...,arg_n) ==> communicate_event(m(arg_1,...,arg_n)).
\end{Verbatim}

For example, the below Java program performs the bubble sort algorithm:\footnote{The program uses the same algorithm provided in www.mathbits.com/MathBits/Java/arrays/Bubble.htm.}
\begin{lstlisting}[frame=single]
public class MySort {
	 public static void main(String[]args)
	 {
		initializeAndSort();
	 }
   public static void setValue(int[]num, int index, int newValue)
	 {
			num[index]=newValue;
	 }
	
	 public static void initializeAndSort()
	 {
		int[] numbers = new int[4];
		setValue(numbers, 0, 20);
		setValue(numbers, 1, 10);
		setValue(numbers, 2, 5);
		setValue(numbers, 3, 1);
		boolean swapped = true; 
		int temp;
		while (swapped==true) {
			swapped=false;
				for (int i = 0; i < numbers.length - 1; i++) {
					if (numbers[i] > numbers[i + 1]) 
					{
						temp = numbers[i]; // swap elements
						setValue(numbers,i,numbers[i+1]);
						setValue(numbers, i+1, temp);
						swapped = true; 
					}
				}
		}
	 }
}
\end{lstlisting}

As seen in Figure \ref{fig:javaann}, each Java method is annotated by linking it to an object. Similarly, the panel is populated with the corresponding visual aspects of the chosen object (node in the shown example).
The user can enter for every parameter a constant value. The value of the parameter could also be linked to one of the arguments of the annotated method through using the built-in function \verb+valueOf/1+.

\begin{figure}[!ht]%
\includegraphics[width=110mm]{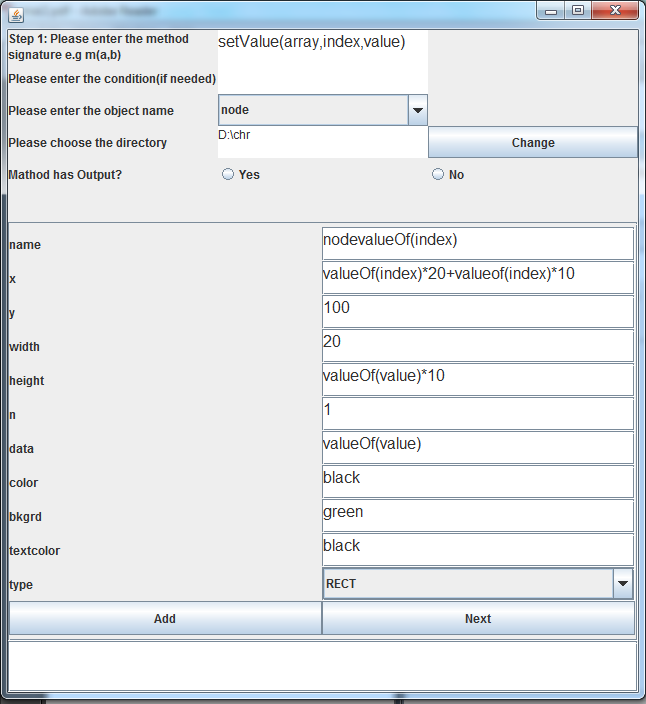}%
\caption{Annotating a Java method}%
\label{fig:javaann}%
\end{figure}

In the previous example, each time the method \verb+setValue/3+ is called, a node is generated. The position of the node is calculated through the argument $index$. Its height is a factor of the argument $value$.
Thus every time a value changes inside the array, the corresponding node is produced. This leads to visualizing the array and the changes happening to it.
The produced step-by-step animation is shown in Figure \ref{fig:animBubbleJava}.

Animating Java programs in the shown way is a generic one since it does not restrict the user to any specific visual data structure. It uses Jawaa providing the basic visual structures which could be used to do target animations. Such animations could also be prepared before or done while executing the code. In addition, since the system is built through SWI-Prolog, this makes it portable.

\begin{figure}
\centering
\begin{tabular}{ccc}
\subfloat{\includegraphics[width=20mm]{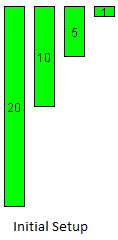}} & 
\subfloat{\includegraphics[width=20mm]{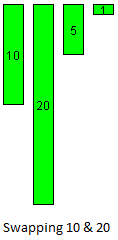}} &
{\subfloat{\includegraphics[width=20mm]{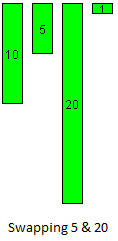}}} \\
{\subfloat{\includegraphics[width=20mm]{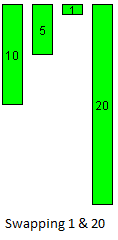}}} &
\subfloat{\includegraphics[width=20mm]{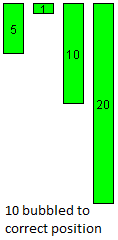}} & 
\subfloat{\includegraphics[width=20mm]{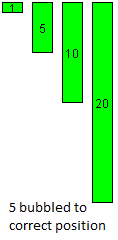}}\\
\end{tabular}
\caption{Bubble Sort Animation}
\label{fig:animBubbleJava}
\end{figure}

\subsection{Building Platforms to teach Mathematics through Animation}
In \cite{DBLP:conf/ruleml/SharafAF16}, the concept of animating programs was used to build a platform to practice different mathematical concepts. The user is offered with different ways to specify what the mathematical rule is.
The specified definition is then transformed into \textsf{CHR} programs.
For a simple rule, the user specifies a name for the rule, its input(s) and output as seen in Figure \ref{fig:simplerulefigure}.
\begin{figure}[!ht]
\centering
  \subfloat[Simple rule: Homepage]{\label{fig:fig2a}\includegraphics[width=60mm]{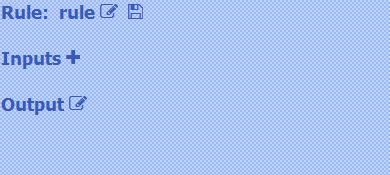}}
  ~
 \subfloat[Adding a new input]{\label{fig:fig2b}\includegraphics[width=60mm]{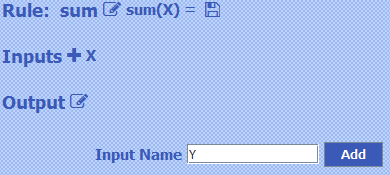}}
 \\
 \subfloat[Editing output]{\label{fig:fig2c}\includegraphics[width=60mm]{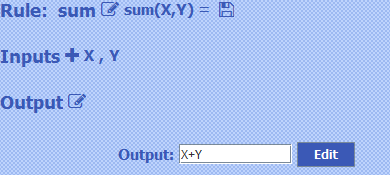}}
 ~
  \subfloat[Summation rule defined]{\label{fig:fig2d}\includegraphics[width=60mm]{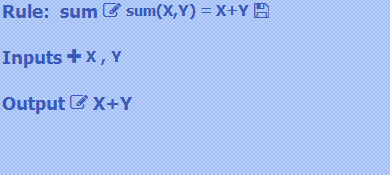}}
  \caption{Simple rules: inputs and outputs.}
	\label{fig:simplerulefigure}
\end{figure}
The user chooses how a number should be visualized. A number $n$ could be linked with any Jawaa object. It could be also linked to a number ($m$) of Jawaa objects as shown in Figure \ref{fig:linkmaths}. In that example, each number $n$ is linked to $n$ Jawaa image objects. Each image object has an x-coordinate and a y-coordinate and a path. An image object displays the image in the specified path.
\begin{figure}[!ht]
\centering
\includegraphics[width=90mm]{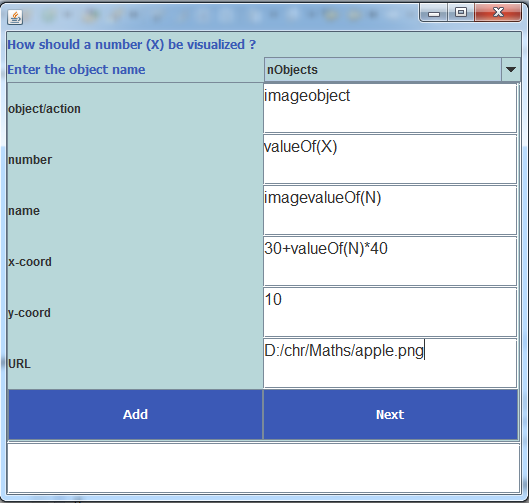}
  \caption{Link a number to a visual object}
  \label{fig:linkmaths}
\end{figure}

The parameter $valueOf\left(N\right)$ has a value of $0$ for the first object, $1$ for the second, $\ldots$, etc.
Two applications were built using the animation.
The first one is shown in Figure \ref{fig:figmathapp1}. The inputs are shown to the user. A number $x$ was annotated with  $x$ image objects with a path to an apple image. Thus each number $x$ is shown as $x$ apples. As seen from the figure, the actual y-coordinate shown to the user is a multiple of the value entered while annotation. Thus, every number is shown on one line. Users then click on ``Add Output'' to formulate their designated output. Since the output is a number, it is visualized in the same way. At any point, the user can choose to check their answer to get a corresponding message.
More details of the input generation is shown in \cite{DBLP:conf/ruleml/SharafAF16}. It is done randomly. However, it can also take into account some constraints/boundaries entered by the user.
\begin{figure}[!ht]
\centering
  \subfloat[Inputs]{\label{fig:fig5a}\includegraphics[width=60mm]{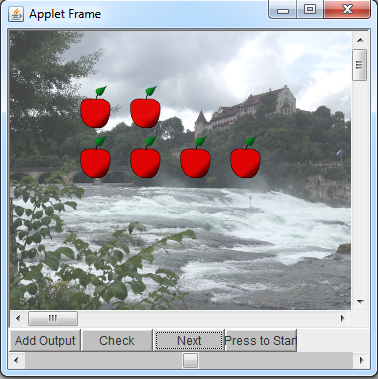}}
  ~
 \subfloat[Editing output I]{\label{fig:fig5b}\includegraphics[width=60mm]{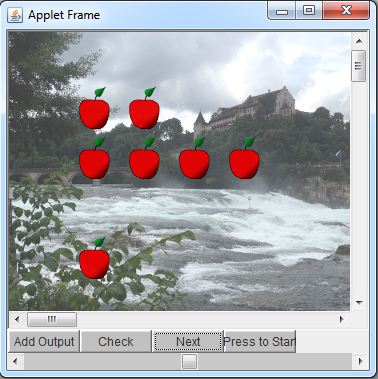}}
 \\
 \subfloat[Editing output II]{\label{fig:fig5c}\includegraphics[width=60mm]{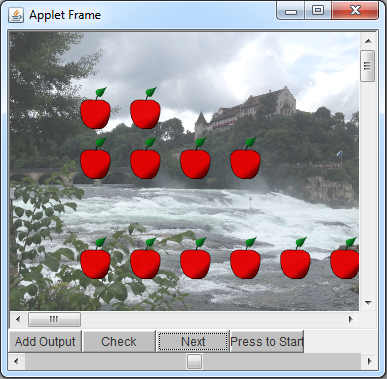}}
  \caption{Quiz 1}
  \label{fig:figmathapp1}
\end{figure}
Another possible animation (shown in Figure \ref{fig:figmaths2}) is to:
\begin{enumerate}
    \item Link every input number with a normal Jawaa circular node. The text inside
the node is its value. Its background is blue.
\item Link the output with a random number of nObjects displaying a group of
nodes. Each node is placed in a random position. The text inside each node
is also a random number. Such nodes have green backgrounds.
\item Link the output with a Jawaa circular node with the name (\verb+jawaanodeout+)
displaying the actual output of the rule. It is also placed at a random position.
Its background is green as well. The output thus has two groups of nodes associated to it.
\item Add an annotation rule linking the output constraint with an \verb+onclick+ command
for the object \verb+jawaanodeout+. Once it is clicked, the \verb+changeParam+
command is activated changing its color to red. Thus the only node whose color changes when clicked is the correct output node.
\end{enumerate}
\begin{figure}[!ht]
\centering
  \subfloat[Randomly placed nodes]{\label{fig:fig6a}\includegraphics[width=75mm]{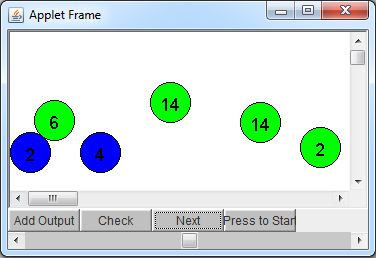}}
 \\\subfloat[Highlighted node after clicking]{\label{fig:fig6b}\includegraphics[width=75mm]{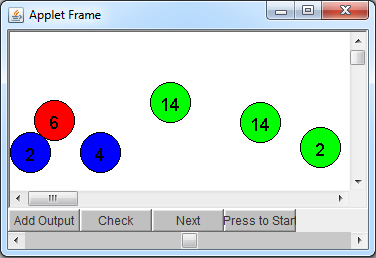}}
\caption{Quiz 2}
  \label{fig:figmaths2}
\end{figure}

\subsection{Animating Cognitive Models}
Another application of animating \textsf{CHR} programs was using it to animate cognitive architectures and the execution of cognitive models through them as shown in \cite{DBLP:conf/gcai/Nada}.
A cognitive architecture includes the basic aspects of any cognitive agent. It consists of different correlated modules \cite{CBO9780511816772A008}. 
Adaptive Control of Thought-Rational (ACT-R) is a well-known cognitive architecture. It was developed to deploy models in different fields including, among others, learning, problem solving and languages \cite{anderson_atomic_1998,Anderson04anintegrated}.
Through animating the execution, users get to see at each step not only details about the model. However, they are also able to visually see the modules of the architecture at all steps of execution.
The idea is to use the previously proposed \textsf{CHR} implementation of ACT-R \cite{Daniel:thesis}. The implementation represented the ACT-R architecture and how models are executed through \textsf{CHR} constraints and rules.
This section shows how animation was achieved through using annotation rules.

ACT-R \cite{anderson_atomic_1998,Anderson04anintegrated} is a cognitive architecture used to execute different type of models simulating human behavior.
ACT-R has different modules integrated together to simulate the different components of the mind that have to work together to reach plausible cognition. ACT-R has different types of buffers holding pieces of information/chunks. At each point in time, a buffer can have one piece of information.
A module can only access the contents of a buffer through issuing a request that is handled by the procedural module.
ACT-R chooses at each step an applicable production rule for execution.
In its \textsf{CHR} implementation, execution is triggered by the constraint \verb+run/0+. $run$ is associated with multiple annotation rules to produce the initial view of the architecture as shown in Figure \ref{fig:basicmodules}.
Figure \ref{fig:basicmodules} shows the basic ACT-R modules.
\begin{figure}[!ht]
\centering
  \includegraphics[width=80mm]{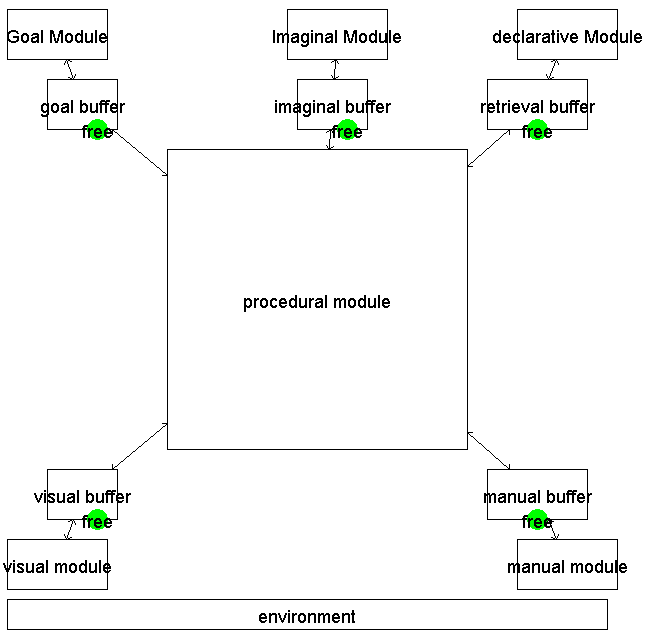}
	\caption{A snapshot of the first panel the user gets. It shows the basic modules of ACT-R as described in \cite{Daniel:thesis,inbookcog}}
	\label{fig:basicmodules}
\end{figure}
\subsubsection{Declarative Module}
\(\)\\This module holds the information humans are aware of. Its corresponding buffer is referred to as the retrieval buffer. Information is represented as chunks of data. Each chunk has a type. 
In the \textsf{CHR} implementation, the information chunks are represented with the \textsf{CHR} constraints: \verb+chunk/2+ and \verb+chunk_has_slot/3+.
\verb+chunk(N,T)+ {represents a chunk named} $N$ with the type $T$. {For example}, \verb+chunk(d,count_order)+ represents a chunk named $d$ with type $count\_order$. On the other hand, \verb+chunk_has_slot/3+ represents the values in the slots of a chunk. {The two constraints} \verb+chunk_has_slot(d,first,3)+ and  \verb+chunk_has_slot(d,second,4)+
represent that chunk $d$ has the values $3$ and $4$ in the slots $first$ and $second$ respectively. This represents the information that $3$ is less than $4$.

\subsubsection{Buffer System}
\(\)\\
\begin{figure}[!ht]
\centering
\subfloat[The retrieval buffer is empty. Its state is busy.]{
        \label{subfig:fig1}
        \includegraphics[width=95mm]{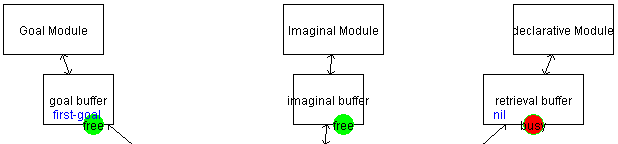} } 
				
\subfloat[As a result of a request made, the retrieval buffer has $c$ and its state is free.]{
        \label{subfig:fig2}
        \includegraphics[width=95mm]{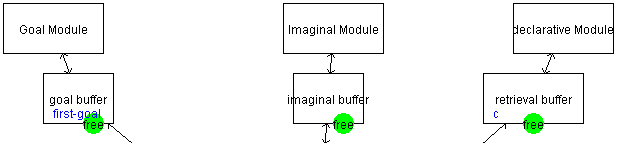} } 
				
\subfloat[The buffer is performing a request. Its state is busy.]{
        \label{subfig:fig3}
        \includegraphics[width=95mm]{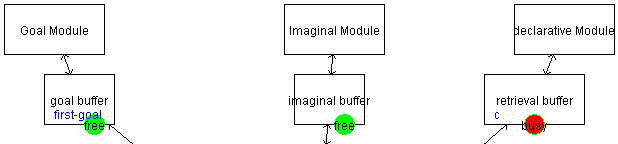}} 
				
\subfloat[The buffer contains $d$ and it is free again]{
        \label{subfig:fig4}
        \includegraphics[width=95mm]{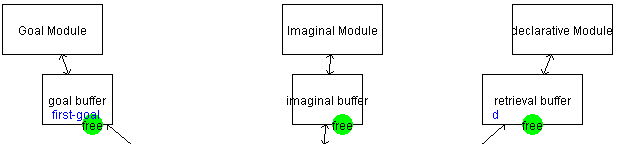}} 
\caption{Buffer System Visualization}
\label{fig:buffers}
\end{figure}
The \textsf{CHR} constraint \verb+buffer(B,C)+ represents the fact that buffer \verb+B+ is holding the chunk $C$. The state $S$ of a buffer $B$ is represented by the constraint \verb+buffer_has_state(B,S)+.
A buffer has one of three states: either $free$, $busy$ or $error$. The buffer is $busy$ while completing a request. 
The state of a buffer is set to $error$ if the request was not successful. 
As shown in Figure \ref{fig:buffers}, each buffer is associated with several visual features.
First of all, \verb+buffer(B,C)+ is annotated with a textual object showing the content of the buffer. Each \verb+buffer(B,C)+ also produces a circular colored (initially green) node. \verb+buffer_has_state(B,S)+ is associated with annotation rules that change the color of the circular node according to the value of $S$.

\subsubsection{Procedural Module: Timing \& Prioritizing Actions in ACT-R}
\(\)\\
\begin{figure}
\captionsetup[subfloat]{farskip=2pt,captionskip=1pt}
  \centering 
  \subfloat[][]{\includegraphics[width=55mm]{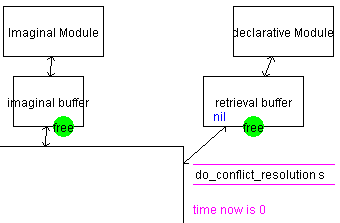} \label{fig:step2}}%
  \quad 
  \subfloat[][do\_conflict\_resolution is dequeued. The rule \emph{start} is a matching rule chosen to be applied.]{\includegraphics[width=55mm]{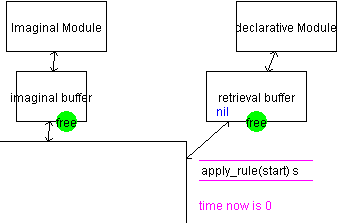} \label{fig:step3}} 
	\quad 
  \subfloat[][apply\_rule(start) is dequeued]{\includegraphics[width=55mm]{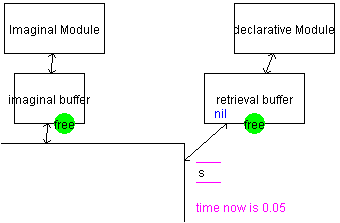} \label{fig:step4}} 
\quad 
  \subfloat[][The action events of \emph{start} are added to the queue.]{\includegraphics[width=55mm]{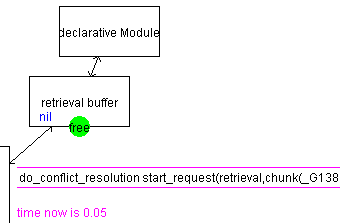} \label{fig:step6}}%
	\quad
  \subfloat[][Starting requesting the declarative module]{\includegraphics[width=55mm]{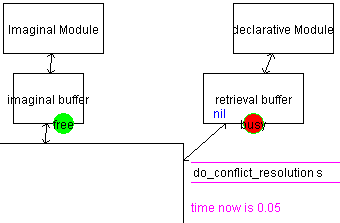} \label{fig:step8}}%

	\subfloat[][]{\includegraphics[width=55mm]{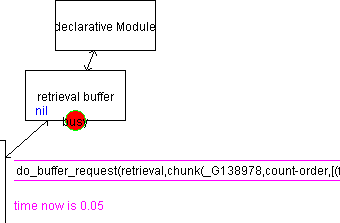} \label{fig:step9}}%
	\quad
  \subfloat[][]{\includegraphics[width=55mm]{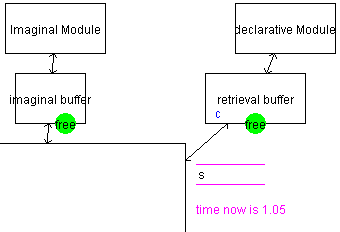} \label{fig:step12}}%
	\caption{Animating Scheduling}
	 \label{fig:cont4}
	\end{figure} 
The current timing of the ACT-R system is represented in its \textsf{CHR} implementation by a constraint \verb+now/1+. 
The initial constraint $run$ produces a textual object. $now$ is associated with an annotation rule to change the value of the text.
The implementation has a central scheduling unit which keeps track of the events to be performed and their timings.
The \textsf{CHR} implementation uses a priority queue to keep track of the actions. The scheduler removes the first event in the queue from time to time.
Event $A$ precedes $B$ in the queue if $A$ has less timing than $B$. If they have the same timings, $A$ precedes $B$ if it has a higher priority.
The order of elements in the queue is represented by the constraint \verb+->/2+. \verb+A -> B+ means that $A$ precedes $B$.
The initial constraint $run$ produces a Jawaa queue object.
The constraint $A->B$ is associated with a rule that fires an action to insert in the queue $B$ after $A$. An example is shown in Figure \ref{fig:cont4}.

\section{Related Work}
In general, algorithm animation or software visualization produces abstractions
for the data and the operations of an algorithm. The different states of
the algorithm are represented as images that are animated according to the
different interactions between such states \cite{bookintro}. In \cite{DBLP:journals/vlc/HundhausenDS02}, some of the scenarios
in which Algorithm Visualization (AV) could be used were discussed. Such
scenarios include using AV technologies in lecture slides, or in practical laboratories,
or for in-class discussions, or in assignments where students could
for example produce their own visualizations or in office hours for instructors
to find bugs quickly. Moreover, such visualizations could be useful for
debugging and tracing the implementations of different algorithms.

In \cite{DBLP:journals/vlc/HundhausenDS02}, a meta-study of 24 experimental studies was performed. Through
the analysis done it was found that such visualizations are educationally
effective.
As introduced in \cite{DBLP:journals/vlc/HundhausenDS02}, eleven studies show significant difference between a
group of study using some configuration of AV technology and another group
that is either not using AV technology at all or using a different configuration.
In \cite{DBLP:journals/vlc/HundhausenDS02}, after performing the analysis, it was found out that the biggest impact
on educational effectiveness results from how students use AV technology
rather than what they see.

In \cite{developing}, a visualization tool for Artificial Intelligence (AI) searching algorithms
was introduced. Moreover, a study was performed in order to know
the effect of using the tool on students. According to \cite{developing}, students who have better visual representation of the changing data structures understood
the algorithms better. The results of the performed study suggest that using
visualization tools could help students in understanding such searching
algorithms.

Several systems and steps towards visualizing the execution of the different
types of algorithms have been made.
For example in \cite{Baec:98}, a 30-minute film designed to teach nine sorting algorithms
was used for algorithm animation as an alternative to doing such
demonstration manually on a board.
Some systems have also provided their users with the possibility to visualize
the execution of different algorithms. For example, XTANGO \cite{Stasko:xtango} is a general
purpose animating system that supports the development of algorithm
animations where the algorithm should be implemented in C or another language
such that it produces a trace file to be read by a C program driver.
The important operations and events that should be highlighted during the
execution of the program should also be portrayed.
Another system is BALSA \cite{Brown:1984:SAA:800031.808596}, in which the notion of interesting events was
used where the animator and the algorithm designer have to agree on a plan
for visualization to identify interesting events that could change the visualized
images.
Zeus \cite{DBLP:conf/vl/Brown91}, also uses the notion of interesting events and annotates the algorithm
with markers that identify the basic operations. In Zeus, its preprocessor
Zume, reads the event specifications and generates definitions for
algorithm and view classes.

In \cite{238854}, and \cite{Eisenstadt:1987:GDT:1625015.1625030}, visualization of logic programs was presented. In \cite{238854},
logic programs were graphically represented using a variation of cyclic AND/OR
graphs. The structure of logic programs was represented through static
graphs that show and correspond to the structure of the source code. Dynamic
graphs show steps of the solution. Finally, a set of binding dependency
graphs were used for showing how the values of the different variables were
generated.
In \cite{Eisenstadt:1987:GDT:1625015.1625030}, Augmented AND/OR trees (AORTA) were used for having a tracing
and debugging facility for Prolog. In the produced graphs all subgoals and
bindings are shown.

On the other hand some visualization tools were offered for constraint
programs specifically. For example, in \cite{debuggingcp}, the tool Grace, was offered as a
constraint tracing environment on top of ECLiPSe \cite{DBLP:citeseer_oai:CiteSeerX.psu:10.1.1.53.5817} where the focus is on
the search space and the domains. The main target of Grace is CLP(FD) programs specifically programs that use labeling of finite domains and backtracking
search. The FD variables and their domains are shown to users. In
addition, a variable stack is used to display the current position in the search
space where each row corresponds to a labeled variable. The variable's domain
in addition to the depth in the search space and the variable position
are also shown through the rows. The displayed domains also differentiate
between the current values, the values that are still be to be tried and the
values that have been tried and failed.

The Oz Explorer provided through \cite{ozexp} also supports the development of
constraint programs. This visual constraint programming tool is provided
for the language Oz \cite{kerneloz,Vol1000} which is a concurrent constraint language. The
visualized object in this tool is the search space. The search tree is visualized
as it is explored. Nodes carry information about corresponding constraints.
Users could interact with the visualized tree through expanding and collapsing
different parts.

Similarly, in \cite{Simonis:2000:SV:646018.678425}, the main focus is on search trees and the offered tool is used
for debugging and analyzing trees generated from different finite domain constraint
programs. The tool offers different views for the user through which
he/she is able to build up information regarding the variables and their domains,
in addition to the search tree and the constraints. Users are offered
information regarding search and constraint propagation. They could also
view and analyze the change of constraints and variables along a path of the
tree.

\section{Conclusions}
In conclusion, the paper presented a formalization for embedding animation features into \textsf{CHR} programs.
The new extension, $CHR^{vis}$ is able to allow for dynamic associations of constraints and rules with visual objects. The annotation rules are thus activated on the program's execution to produce algorithm animations. 
Although the idea of using interesting events was introduced in earlier work, it was (to the best of the authors' knowledge) never formalized before. In fact, no operational semantics for animation was proposed before. The paper offered operational semantics for $CHR^{vis}$. It thus provides a foundation for formalizing the animation process in general and for \textsf{CHR} programs in particular.
In the future, with the availability of formal foundations through $\omega_{vis}$ and $\omega_{vis_{r}}$, the possibility of using $CHR^{vis}$ as the base of a pure a visual representation for \textsf{CHR} should be investigated. 

\bibliographystyle{ACM-Reference-Format}
\bibliography{Master}

\end{document}